\documentclass[sigconf]{acmart}

\pdfoutput=1 

\let\ACMmaketitle=\maketitle
\renewcommand{\maketitle}{\begingroup\let\footnote=\thanks 
	\ACMmaketitle\endgroup}

\renewcommand\footnotetextcopyrightpermission[1]{} 

\acmConference[Technical Report]{Technical Report}{2021}{}
\startPage{1}

\setcopyright{none}

\usepackage{thmtools} 
\usepackage{thm-restate}
\usepackage{soul}

\usepackage{color, colortbl}

\usepackage{moresize}

\usepackage{dblfloatfix}

\usepackage{rotating}
\usepackage{graphicx}
\usepackage{multirow}
\usepackage{geometry}
\usepackage{amsfonts}
\usepackage{amsmath}
\usepackage{mathtools}
\usepackage{wrapfig}
\usepackage{listings}
\usepackage{graphicx}
\usepackage[font={small}]{caption}
\usepackage{tabularx}
\usepackage{fontawesome}
\usepackage{pifont}
\usepackage{enumitem}
\usepackage{algpseudocode}
\usepackage[font={sf,footnotesize}]{subfig} 
\usepackage{amsthm}
\usepackage{placeins}
\usepackage{bm}
\usepackage[framemethod=tikz]{mdframed}

\usepackage{algorithm}
\usepackage{algpseudocode}
\usepackage{algorithmicx}
\usepackage{verbatim}
\usepackage{xr}
\usepackage{float}
\usepackage{listings}

\floatstyle{plain}
\newfloat{lstfloat}{htbp}{lop}
\floatname{lstfloat}{Listing}

\newcommand\fw{0.32}

\newtheorem*{corollary*}{Corollary}
\newtheorem{corollary}{Corollary}
\newtheorem{lma}{Lemma}

\newcommand\Intuition[1]{\textbf{Intuition.} \emph{#1}}
\newcommand\Note[1]{\textbf{Note.} \emph{#1}}

\definecolor{darkgreen}{rgb}{0.0, 0.5, 0.13}

\DeclareMathOperator*{\argmax}{arg\,max}
\DeclareMathOperator*{\argmin}{arg\,min}

\definecolor{darkgrey}{RGB}{70,70,70}
\definecolor{lightgrey}{RGB}{200,200,200}

\usepackage{cleveref}
\usepackage{xspace}
\usepackage[utf8]{inputenc}
\crefname{section}{§}{§§}
\Crefname{section}{§}{§§}

\newcommand*{\affmark}[1][*]{\textsuperscript{#1}}

\DeclareSymbolFont{matha}{OML}{txmi}{m}{it}
\DeclareMathSymbol{\varS}{\mathord}{matha}{83}

\colorlet{hlcolor}{yellow!20}

\DeclareRobustCommand{\hltext}[1]{{#1}}

\mdfdefinestyle{graybox}{
	splittopskip=0,%
	splitbottomskip=0,%
	frametitleaboveskip=0,
	frametitlebelowskip=0,
	skipabove=0,%
	skipbelow=0,%
	leftmargin=0,%
	rightmargin=0,%
	innerrightmargin=0mm,%
	innerleftmargin=0mm,%
	innertopmargin=1mm,%
	innerbottommargin=0mm,%
	roundcorner=0mm,%
	backgroundcolor=hlcolor}

\newcommand{\macb}[1]{\textbf{\textsf{#1}}}

\usepackage{booktabs} 
\usepackage{makecell}

\usepackage{pifont}

\begin{document}
	
	\newcommand{\conflux}{CO$\mathit{nf}$LUX\xspace}
	\newcommand{\chol}{CO$\mathit{nf}$CHOX\xspace}
	\newcommand{\xparting}{\mbox{$X$-Partitioning}\xspace}
	\newcommand{\xpart}{\mbox{$X$-partition}\xspace}

\title[Near-I/O-Optimal Matrix Factorizations]{On the 
Parallel I/O Optimality of Linear Algebra Kernels: 
Near-Optimal Matrix Factorizations}         

\author{
	Grzegorz Kwasniewski\affmark[1], 
	Marko Kabic\affmark[1]\affmark[2], 
	Tal Ben-Nun\affmark[1], 
	Alexandros Nikolaos Ziogas\affmark[1],
	Jens Eirik Saethre\affmark[1], 
	André Gaillard\affmark[1],
Timo Schneider\affmark[1], 
Maciej Besta\affmark[1],
Anton Kozhevnikov\affmark[1]\affmark[2],
Joost VandeVondele\affmark[1]\affmark[2],
 Torsten Hoefler\affmark[1]\\
	{\vspace{1em}\affmark[1]Department of Computer Science, 
	ETH Zurich, 
\affmark[2]Swiss National Computing Center
\vspace{1em}}
}

\renewcommand{\shortauthors}{G. Kwasniewski et al.}

\begin{abstract}
Matrix factorizations are among the most important building 
blocks of scientific computing. However, state-of-the-art libraries 
are not communication-optimal, underutilizing current 
parallel architectures.  We present novel algorithms for 
Cholesky and LU factorizations that utilize an
asymptotically communication-optimal 2.5D decomposition. 
We first establish a theoretical framework for deriving parallel 
I/O lower bounds for linear algebra kernels, 
and then utilize its insights to derive Cholesky and LU 
schedules, both communicating $N^3/(P \sqrt{M})$ 
elements per processor, where M is the local memory size. 
The empirical results match our 
theoretical analysis: our implementations communicate 
significantly less than Intel MKL, SLATE, and the
asymptotically communication-optimal CANDMC and CAPITAL libraries.
Our code outperforms these state-of-the-art libraries
in almost all tested scenarios, with matrix sizes ranging from 2,048 to 
524,288 on up to 512 CPU nodes of the Piz Daint supercomputer, 
decreasing the time-to-solution by up to three times. Our code is 
ScaLAPACK-compatible and 
available as an open-source library.
\end{abstract}

\maketitle

\section{Introduction}
\label{sec:intro}

Matrix factorizations, such as LU and Cholesky 
decompositions, play a crucial role in many 
scientific 
computations~\cite{rectangularML, 
meyer2000matrix,krishnamoorthy2013matrix}, 
and their performance can dominate the 
overall runtime of entire applications~\cite{rpa}. 
Therefore,
accelerating these routines is of great 
significance for numerous domains~\cite{joost, cp2k}. 
The ubiquity and importance 
of LU factorization is even reflected by the fact that it is 
used to rank top 
supercomputers worldwide~\cite{TOP500_HPL}.

Since the arithmetic complexity of matrix factorizations is $\mathcal{O}(N^3)$ 
while the input size is $\mathcal{O}(N^2)$, these kernels are traditionally 
considered
compute-bound. However, the end of Dennard 
scaling~\cite{dennard1974design} puts 
increasing pressure on data movement minimization, as the cost of moving data 
far exceeds its computation cost, both in terms of power and 
time~\cite{kestor2013quantifying,padal}.
Thus, deriving 
algorithmic I/O lower bounds is a subject of both theoretical 
analysis~\cite{general_arrays, redblue, IronyMMM} and practical value for
developing I/O-efficient schedules~\cite{maciejBC, edgarTradeoff, 
choleskyQRnew}.

While asymptotically optimal matrix factorizations were proposed, among others, 
by Ballard et al.~\cite{cholesky1} and Solomonik et al.~\cite{2.5DLU, 
choleskyQRnew}, we observe two major challenges with the existing approaches: 
First, the presented algorithms are only asymptotically optimal: the I/O cost 
of these
proposed parallel algorithms can be as high as 7 times the lower bound for 
LU~\cite{2.5DLU} and up to 16 times for Cholesky~\cite{choleskyQRnew}. This 
means that they communicate less than ``standard'' 2D 
algorithms like ScaLAPACK~\cite{scalapackLayout} only for 
almost prohibitively large numbers of 
processors --- e.g., according to the LU cost 
model~\cite{2.5DLU}, it requires more 
than 15,000 processors to 
communicate less than an optimized 2D algorithm. 
Second, their time-to-solution performance can be worse than highly-optimized, 
existing 2D-parallel libraries~\cite{choleskyQRnew}.

To tackle these challenges, we first provide a 
\emph{general} 
method
for deriving \emph{precise} I/O lower bounds of Disjoint 
Array Access 
Programs 
(DAAP) --- a broad range
of programs composed of a sequence
of statements enclosed in an arbitrary number of nested loops.
We then illustrate the applicability of our 
framework to derive 
parallel I/O lower bounds of Cholesky and LU 
factorizations: $\frac{1}{3}\frac{N^3}{P \sqrt{M}}$ and 
$\frac{2}{3}\frac{N^3}{P \sqrt{M}}$ elements, respectively, 
where $N$ is the  
matrix size, $P$ is the number of processors, and $M$ is the 
local memory size.

	\begin{figure}
	\centering
	\subfloat
	{\hspace{-1.6em}
		\includegraphics[width=0.283 \textwidth]
		{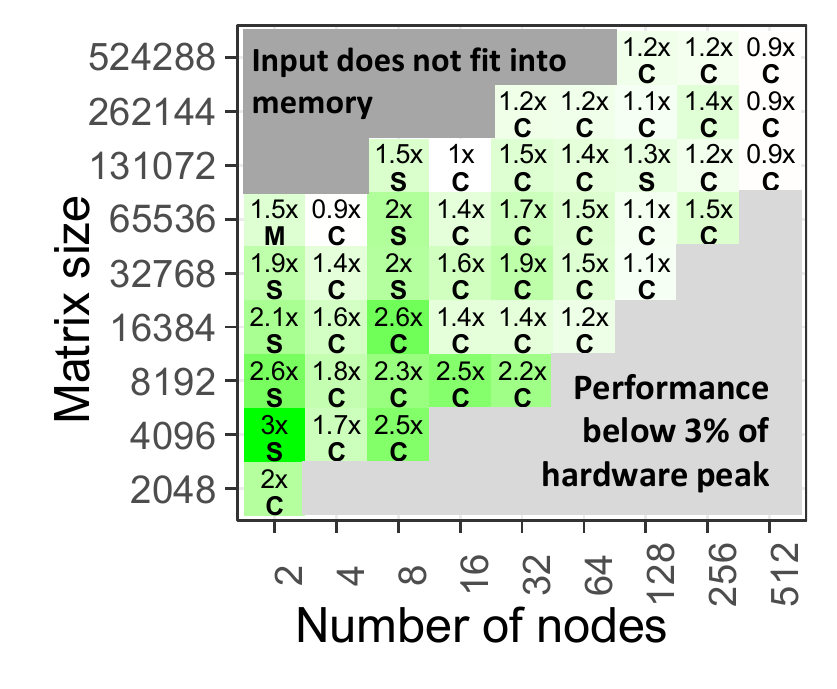}}
	%
	\subfloat
	{		
		\hspace{-1.0em}
		\includegraphics[width=0.234 \textwidth]
		{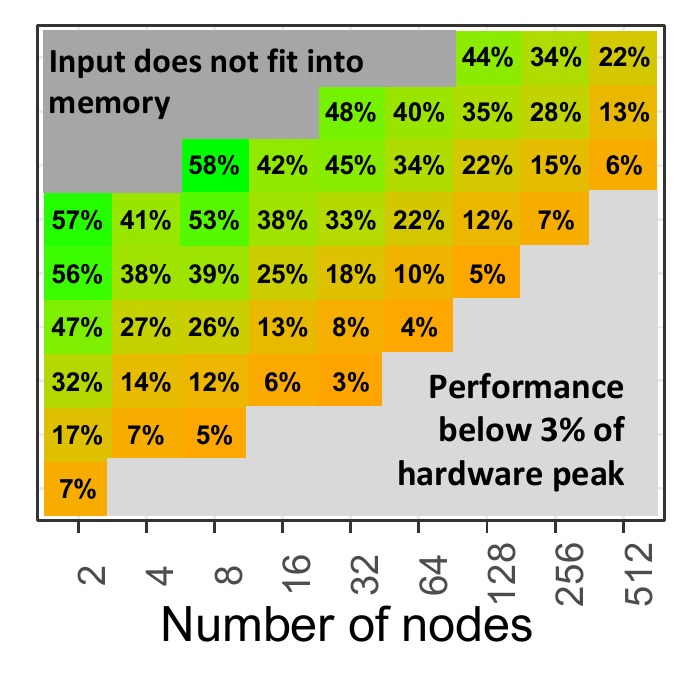}}
	\caption{
		{{\textbf{Left:} measured runtime speedup of \conflux 
				vs. 
				fastest state-of-the-art library 
				(S=SLATE~\cite{slate}, 
				C=CANDMC~\cite{candmc}, 
				M=MKL~\cite{mkl}). \textbf{Right:} \conflux's 
				achieved \% of 
				machine peak 
				performance. } }
	}
	\label{fig:heatmaps_lu}
\end{figure}

\begin{figure*}[h]
	\vspace{2em}
	\includegraphics[width=2.1\columnwidth]{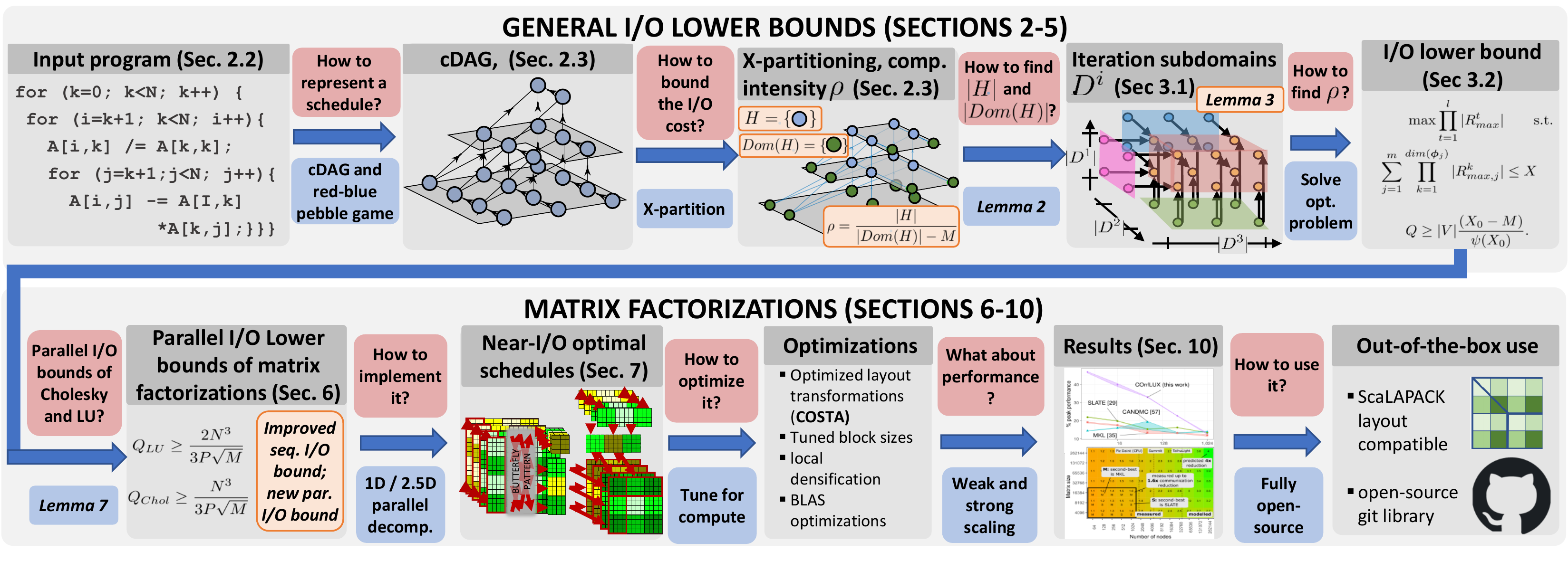}
	\caption{From the input program through the I/O lower bounds to 
		communication-minimizing parallel schedules and high performing 
		implementations. In this paper, we mainly focus on the Cholesky and LU 
		factorizations. The 
		proofs of the lemmas presented in this work can be found 
		in the AD/AE appendix.}
	\label{fig:sc_flow}
\end{figure*}

Moreover, we use the 
insights from deriving
the above lower bounds to develop \conflux and \chol, near 
communication-optimal parallel LU and Cholesky factorization algorithms
that minimize data movement across the 2.5D processor decomposition. For 
LU factorization, to further reduce the latency and bandwidth cost, we use 
a 
row-masking 
tournament pivoting strategy resulting in a communication requirement of  
$\tfrac{N^3}{P \sqrt{M}} + \mathcal{O}(\tfrac{N^2}{P})$ 
elements per
processor, where the leading order term is only 1.5 times the lower 
bound. 
Furthermore, to secure high performance, we 
carefully 
tune block sizes and communication routines to maximize the
efficiency of local 
computations such as \texttt{trsm} (triangular solve) and \texttt{gemm} (matrix 
multiplication).

We measure both communication 
volume and achieved performance of
\conflux and \chol
and compare them to state-of-the-art libraries: 
a vendor--optimized Intel MKL~\cite{mkl}, SLATE~\cite{slate} 
(a recent library targeting exascale systems), as well as
CANDMC~\cite{candmc, candmccode} and 
CAPITAL~\cite{choleskyQRnew, choleskycode} 
(codes based on the
asymptotically optimal 2.5D decomposition).
In our experiments on the Piz Daint supercomputer, we measure 
up to 1.6x 
communication reduction 
compared to the second-best implementation. Furthermore, our 2.5D 
decomposition communicates asymptotically less than SLATE and MKL, with even 
greater expected benefits on exascale machines. Compared to the 
communication-avoiding 
CANDMC library with I/O cost of $5N^3/(P\sqrt{M})$ elements~\cite{2.5DLU}, 
\conflux communicates five times less.
Most importantly, \emph{our implementations outperform all compared libraries 
in almost all scenarios}, both for strong and weak scaling, reducing the 
time-to-solution by up to three times compared to the second best performing 
library (Figure~\ref{fig:heatmaps_lu}).

\noindent
In this work, we make the following contributions:
\begin{itemize}[leftmargin=*]
	\item A general method
	for deriving parallel I/O lower bounds of a broad
	range of linear algebra kernels.
	%
	%
	\item \conflux and \chol, provably near-I/O-optimal parallel 
	algorithms for LU and Cholesky factorizations, with their full 
	communication volume analysis.
	\item Open-source and fully ScaLAPACK-compatible implementations of our 
	algorithms that outperform existing state-of-the-art libraries in almost 
	all scenarios.
\end{itemize}

\noindent
A bird's eye view of our work is presented in Figure~\ref{fig:sc_flow}. 

\section{Background}
\label{sec:background}

We now establish the background for our theoretical model 
(Sections~\ref{sec:boundsSingleStatement}-\ref{sec:parredblue}). We use it 
to derive parallel I/O lower bounds for Cholesky and LU factorizations 
(Section~\ref{sec:lu_lowerbound}) that will guide the design of our 
communication-minimizing implementations (Section~\ref{sec:conflux}).
\subsection{Machine Model}
\label{sec:machineModel}

To model algorithmic I/O complexity, we start with a model of a 
sequential machine equipped with a two-level deep memory hierarchy. 
We then outline the parallel machine model.

\noindent \macb{Sequential machine}. A computation is performed on a 
sequential machine 
with a fast memory of limited size and unlimited slow memory. The 
fast memory can hold up to $M$ elements at any given time. 
To perform any computation, all input elements must reside in fast 
memory, and the result is stored in fast memory.

\noindent \macb{Parallel machine}. The sequential model is extended to a 
machine with $P$ processors, each equipped with a private 
fast memory of size $M$. There is no global memory of unlimited 
size --- instead, elements are transferred between processors' fast 
memories. 

\subsection{Input Programs}
\label{sec:inputPrograms}

We consider a general class of programs that operate on 
multidimensional
arrays. Array
elements can be loaded from slow to fast memory, stored from fast to 
slow memory, and computed inside fast memory.  
These elements have \emph{versions} that 
are incremented every time they are 
updated. 
We model the program execution as a computational directed acyclic 
graph (cDAG, details in Section~\ref{sec:pebblegame}), where each vertex 
corresponds to a different version of an array element.
Thus, for a statement $A[i,j] 
\leftarrow f(A[i,j])$, a vertex corresponding to $A[i,j]$ \emph{after} 
applying $f$ is different from a vertex corresponding to $A[i,j]$ 
\emph{before} 
applying $f$.
In a cDAG, this is expressed as an edge from vertex $A[i,j]$ before $f$ to 
vertex $A[i,j]$ after $f$. Initial versions of each element do not have 
any incoming edges and thus form the cDAG inputs. 
\emph{The distinction between elements and vertices} is important for our I/O 
lower bounds analysis, as we will investigate how many vertices
are computed
for a 
given number of loaded vertices.

An input program is a collection of statements $S$ enclosed in loop 
nests, each of the following form (we use the loop nest notation introduced 
by Dinh and 
Demmel{~\cite{demmel2}}):
{\small
	\begin{align}
	\nonumber
	\text{\textbf{for }} \psi^1 \in \mathcal{D}^1,
	\text{\textbf{for }}  \psi^2 \in \mathcal{D}^2(\psi^1),
	\dots ,
	\text{\textbf{for }}  \psi^l \in \mathcal{D}^l(\psi^1, \dots, \psi^{l-1}): 
	\\
	\nonumber
	S: A_0[\bm{\phi_0}(\bm{\psi})] \leftarrow 
	f(A_1[\bm{\phi_1}(\bm{\psi})], A_2[\bm{\phi_2}(\bm{\psi})], \dots, 
	A_m[\bm{\phi_m}(\bm{\psi})]),
	\end{align}
}

\noindent
where (cf.~Figure~\ref{fig:prog_rep} for
a summary) for each innermost loop iteration, statement $S$ is an evaluation of 
some function $f$ 
on 
$m$ 
inputs, 
where every input is an element of array $A_j, j = 1,\dots, 
m$, 
and the result of $f$ is stored to the output array $A_0$.

Each loop has an associated \emph{iteration variable} {$\psi^t$} that 
iterates over 
its domain $\psi^t \in \mathcal{D}^t$. All $l$ iteration 
variables form the \emph{iteration vector} 
$\bm{\psi} = [\psi^1, \dots, \psi^l]$. Array elements are 
accessed by an \emph{access 
	function vector} $\bm{\phi_j} = [\phi_j^1, \dots, 
\phi_j^{dim(A_j)}]$ 
that maps $dim(A_j)$ iteration variables to a \emph{unique} element in 
array {$A_j$} (note that the access function vector is injective). Only 
vertices associated with the newest element versions can be 
referenced. Furthermore, a given vertex can be 
referenced by only one access 
function vector per statement. We refer to this as the \emph{disjoint 
	access property.} 
 The 
	\emph{access dimension} of \mbox{$A_j(\bm{\phi}_j)$}, denoted 
	\mbox{$dim(A_j(\bm{\phi}_j))$}, is the number of 
	different 
	iteration variables present in \mbox{$\bm{\phi}_j$}.
	We call such programs Disjoint Access Array 
	Programs. 
	
	\emph{Example: Consider statement $S1$ of LU factorization 
	(Figure~\ref{fig:prog_rep}). The loop 
		nest depth is $l=2$, with two iteration variables $\psi^1 = $\texttt{ 
		k} 
		and $\psi^2 = $\texttt{ i} forming the iteration vector $\bm{\psi} = $ 
		\texttt{ [k, i]}. For access $A[k,k]$, 
		the access function vector 
		$\bm{\phi}_j = [k,k]$ is a function of only one 
		iteration variable $k$. Therefore, $dim(A_j) = 2$, 
		but 
$dim(A_j(\bm{\phi_j})) = 1$.}

\begin{figure}[t]
	\includegraphics[width=\columnwidth]{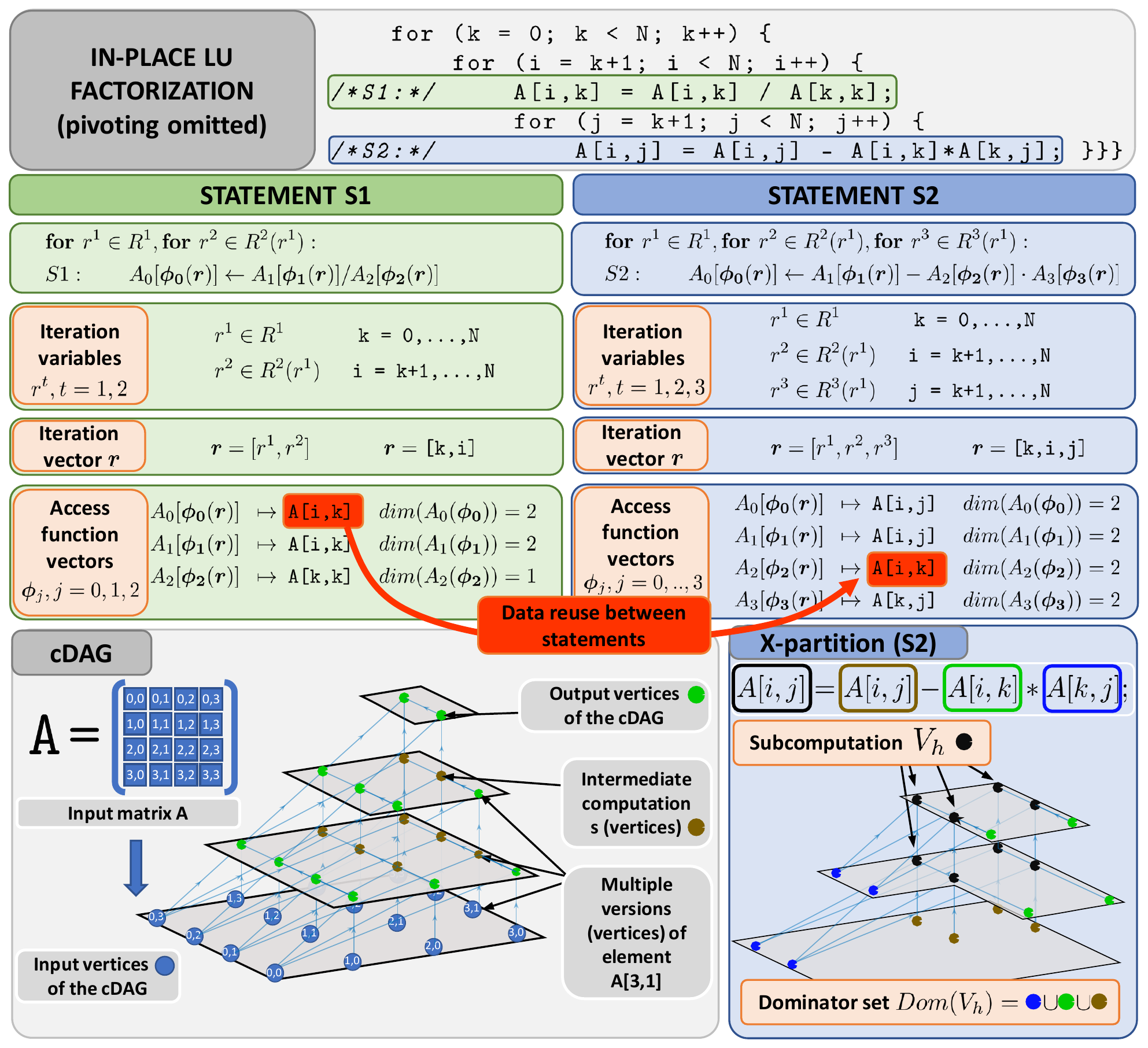}
	
	\caption{In-place LU factorization (for simplicity, 
	no pivoting is 
		performed). The algorithm contains two statements ($S1$ and $S2$), for 
		which we 
		provide key components of our program representation together with the 
		corresponding cDAG for $N=4$. For 
		statement $S2$, we also provide a graphical visualization of a single 
		subcomputation $H$ in its \xpart.}
	
	\label{fig:prog_rep}
\end{figure}

\subsection{{I/O Complexity and Pebble Games}}
\label{sec:pebblegame}
{We now establish the relationship between DAAP and the red-blue pebble game 
	- a powerful abstraction for deriving lower bounds and optimal schedules of 
	cDAG evaluation.}

\subsubsection{cDAG and red-blue pebble game}
We base our computation model on the red-blue pebble game, played on the 
computational directed acyclic graph $G=(V,E)$, as introduced by Hong and 
Kung~\cite{redblue}.
Every vertex $v \in V$ represents the result of a unique 
computation stored in some memory, and a 
directed edge  $(u,v) \in E$ represents a data dependency.
Vertices without any incoming (outgoing) edges are called \emph{inputs} 
(\emph{outputs}). To perform a computation, i.e., to evaluate the value 
corresponding to vertex $v$, all vertices that are direct 
predecessors of 
$v$ must be loaded into fast memory.
The vertices that are currently in fast memory are marked by a red pebble 
on 
the 
corresponding vertex of the cDAG. Since the size of fast 
memory is 
limited, we can never have more than $M$ 
red 
pebbles on the cDAG at any moment.
Analogously, the contents of the slow memory (of unlimited size) is 
represented by an unlimited number of blue pebbles.

\subsubsection{Dominator and Minimum Sets~\cite{redblue}}
\label{sec:redblue} 

For any subset of vertices $H 
\subset V$, a \emph{dominator set} 
$\mathit{Dom}(H)$ is a set such 
that every path in the cDAG from an input vertex to any vertex in 
$H$ must 
contain at least one vertex in 
$\mathit{Dom}(H)$. 
In general, for a 
given $H$, its ${Dom}(H)$ is not 
uniquely 
defined.
The \emph{minimum set} 
$\mathit{Min}(H)$ is the 
set of all vertices in $H$ that do not have any 
immediate successors in $H$.
In this work, to avoid the ambiguity of non-uniqueness of dominator set 
size
(in 
principle, for any subset, its valid dominator set is always the whole 
$V$), 
we will refer to
$\mathit{Dom}_{min}(H)$ as a minimum dominator set, i.e. a dominator set with 
the smallest size.

\vspace{0.5em}
\noindent
\Intuition{
	{One can think of \mbox{$H$}'s dominator set as a 
	set of 
	inputs 
	required to execute subcomputation \mbox{$H$}, 
	and of \mbox{$H$}'s
	minimum set as the output of \mbox{$H$}.
	We use the notions of \mbox{$\mathit{Dom}_{min} \left( H 
	\right)$} and 
	\mbox{$\mathit{Min}\left( H \right)$}
	when proving I/O lower bounds. Intuitively, we bound
	computation ``volume''
	(number of vertices in \mbox{$H$}) by its communication 
	``surface'', comprised 
	of 
	its inputs - vertices in
	\mbox{$\mathit{Dom}_{min}\left(H\right)$}
	- and outputs - vertices in 
	\mbox{$\mathit{Min}(H)$.}}}

\subsubsection{\xparting}
\label{sec:xpart}
Introduced by Kwasniewski et al.~\cite{COSMA},\linebreak \xparting generalizes the 
S-partitioning abstraction~\cite{redblue}.
An \xpart of a cDAG 
is a
collection of $s$ mutually disjoint subsets (referred to as 
\emph{subcomputations})
$\mathcal{P}(X) = \{H_1, 
\dots, H_s\}$,  $\bigcup_{i=1}^s H_i = V$  with two additional 
properties:
\begin{itemize}[leftmargin=0.9em]
	\item $\mathcal{P}(X)$ has no cyclic dependencies between 
	subcomputations.
	\item $\forall H$, $\left|{Dom}_{min}\left( 
	H \right)\right| \le X$ and  $\left| 
	{Min}\left(H\right) \right| \le 
	X$.
\end{itemize}

{For a given cDAG 
	and for any given \mbox{$X > M$}, let 
	\mbox{$\Pi(X)$} denote a set of all its valid \mbox{$X$}-partitions, 
	\mbox{$\mathcal{P}(X) \in \Pi(X)$}. }	
Kwasniewski et al. prove that an I/O optimal schedule of $G$ 
that 
performs $Q$ load and 
store operations has an associated \xpart $\mathcal{P}_{opt}(X) \in \Pi(X)$ 
with size 
$|\mathcal{P}_{opt}(X)| \le \frac{Q + X - M}{X - M}$ {for any \mbox{$X > 
		M$}} (\cite{COSMA}, Lemma 2).

\subsubsection{Deriving lower bounds}

{To bound the I/O cost, we further need to introduce the 
\emph{computational intensity} \mbox{$\rho$}. For each subcomputation 
\mbox{$H_i$, $\rho_i$} is defined as a ratio of the number of computations 
(vertices) in $H_i$ to the number of I/O operations required to pebble $H_i$,  
where the latter is bounded by the size of the dominator set 
\mbox{$Dom(H_i)$}}~\cite{COSMA}. Then, 
the following lemma bounds the number of I/O operations required to pebble a 
given 
cDAG:

\begin{lma}
	\label{lma:compIntensity}
	(Lemma 4 in \cite{COSMA})
	For any constant $X_c$, the number of I/O operations $Q$ required to pebble 
	a 
	cDAG $G=(V,E)$ with $|V| = n$ vertices using $M$ red pebbles is bounded 
	by $Q \ge {n}/{\rho}$, where ${\rho}  = \frac{|H_{max}|}{X_c - M}$ 
	is the 
	maximal 
	computational intensity and $H_{max} = \argmax_{H \in 
		\mathcal{P}(X_c)} |H|$ is the largest 
	subcomputation 
	among all valid \mbox{$X_c$-partitions}.
\end{lma}

\noindent

\section{{\hspace{-0.7em}General Sequential I/O Lower Bounds}}
\label{sec:boundsSingleStatement}

We now present our method for deriving 
the I/O lower bounds 
of a sequential execution of programs in the form defined in 
Section~\ref{sec:inputPrograms}. Specifically, in 
Section~\ref{sec:iobound_singlestatement} we derive I/O 
bounds for programs that contain only a single statement.
In Section~\ref{sec:mult_statements} we extend our 
analysis to capture interactions and reuse between 
multiple statements. 

In this paper, we present only the key lemmas required to establish the lower 
bounds of parallel Cholesky and LU factorizations. However, the method covers a 
much wider spectrum of 
		algorithms. 
		For curious readers, we present all proofs of provided lemmas in the 
		appendix.

We start by stating our key lemma:

\begin{restatable}{lma}{compIntensityPhi}	
	\label{lma:compIntensityPhi}
	If $|H_{max}|$ can be expressed as a closed-form function of $X$, 
	that 
	is if there exists some function $\chi$ such that
	$|H_{max}| = 
	\chi(X)$, 
	then the lower bound on $Q$ can be expressed as
	$$Q \ge n \frac{(X_0 - M)}{\chi(X_0)},$$
	where $X_0 = \argmin_X \rho = \argmin_X \frac{\chi(X)}{X-M}$.
\end{restatable}

\vspace{0.5em}
\noindent
\Intuition{$\chi(X)$ expresses computation ``volume'', while $X$ is its input 
	``surface''. The term $X-M$ bounds the required communication and it comes 
	from 
	the fact that not all inputs have to be 
	loaded (at most $M$ of them can be reused). $X_0$ corresponds to the 
	situation 
	where the ratio of this ``volume'' to the required communication is 
	minimized 
	(corresponding to a highest lower bound).}

\begin{proof}
	Note that Lemma~\ref{lma:compIntensity} is valid for any $X_c$
	(i.e.,  
	for any  
	$X_c$, it gives a valid lower bound). Yet, these bounds are not necessarily 
	tight. 
	As we want to 
	find tight I/O lower bounds, we need to maximize the lower 
	bound. $X_0$ by 
	definition minimizes $\rho$; thus, it maximizes the bound.
	Lemma~\ref{lma:compIntensityPhi} then follows directly from 
	Lemma~\ref{lma:compIntensity} 
	by 
	substituting $\rho = \frac{\chi(X_0)}{X_0 - M}$.
\end{proof}

\Note{}
If function $\chi(X)$ is differentiable and has a global 
minimum, we 
can 
find 
$X_0$ by, e.g., solving the equation $\frac{d\frac{\chi(X)}{X 
		- M}}{dX} = 0$.
The key limitation is that it is not always possible to find $\chi$, that 
is, 
to 
express $|H_{max}|$ solely as a function of $X$. However, for many 
linear 
algebra kernels $\chi(X)$ exists. Furthermore, one can relax this problem 
preserving the correctness of the lower bound, that is, by finding a 
function
$\hat{\chi}: \forall_X \hat{\chi}(X) \ge \chi(X)$.

To find $\chi(X)$, we take advantage of the DAAP structure.
Observe that every computation (and therefore, every compute vertex $v \in V$ 
in the cDAG 
$G=(V,E)$) is executed in a different iteration of the loop nest, and 
thus, 
there is a one-to-one mapping from a value of the iteration vector 
$\bm{\psi}$ to the compute vertex $v$. Moreover, each vertex accessed from 
any of the 
input arrays $A_i$ 
is also associated with some iteration vector value - however, if $dim(A_i) < 
l$, this is a one-to-many relation, as the same input vertex may be used to 
evaluate multiple compute vertices $v$. This is, in fact, the source 
of the data reuse, and exploiting this relation is a key to minimizing the I/O 
cost. If 
for 
all input arrays $A_i$ we have that $dim(A_i) = l$, then for each compute 
vertex $v$, $m$ 
different, unique input vertices are required, there is no data reuse and it 
implies a 
trivial 
computational intensity $\rho = \frac{1}{m}$.

The high-level idea of our method is to \emph{count how many different 
iteration vector values $\bm{\phi}$ can be formed if we know how many 
different values each iteration variable $\phi^1, \dots, \phi^l$ takes}. We now 
formalize this in Lemmas~\ref{lma:rectTiling}-\ref{lma:output_reuse}.

\subsection{\hspace{-0.5em}Iteration vector, iteration domain, access set}
\label{sec:access_sizes}

Each execution of statement $S$ is associated with the 
\emph{iteration vector} value
$\bm{\psi} = [\psi^1, \dots, \psi^l] \in \mathbb{N}^l$ 
representing 
the current 
iteration, 
that is, the values of iteration variables {$\psi^1$}$, \dots, 
${$\psi^l$}. 
Each 
subcomputation $H$ is uniquely defined by all iteration 
vectors' values associated 
with vertices pebbled 
in $H = \{\bm{\psi}_{1}, \dots, \bm{\psi}_{|H|}\}$.
For each iteration variable {$\psi^t$}, $t = 1, \dots, l$, denote the 
set 
of all values that 
{$\psi^t$} 
takes during $H$ as $D^t$.
We define $\bm{D} = [D^1, \dots, D^t] 
\subseteq 
\mathcal{\bm{D}}$ as 
the \emph{iteration 
	domain} of subcomputation $H$.

Furthermore, recall that each input access $A_j[\bm{\phi_j}(\bm{\psi})]$ is 
uniquely defined by $dim(\bm{\phi}_j)$ iteration variables $\psi_j^1, \dots, 
\psi_j^{dim(\bm{\phi}_j)}$. Denote the set of all values each of $\psi_j^k$ 
takes 
during 
$H$ as $D^k_{j}$.
Given $\bm{D}$, we also denote the number of 
different vertices that are 
accessed from each input array {$A_j$} as $|A_j(\bm{D})|$.

We now state the lemma which bounds $|H|$ by the iteration sets' sizes 
$|D^t|$ (the intuition behind the lemma is depicted in Figure~\ref{fig:lemma3}):

\begin{figure}
	\includegraphics[width=1\columnwidth]{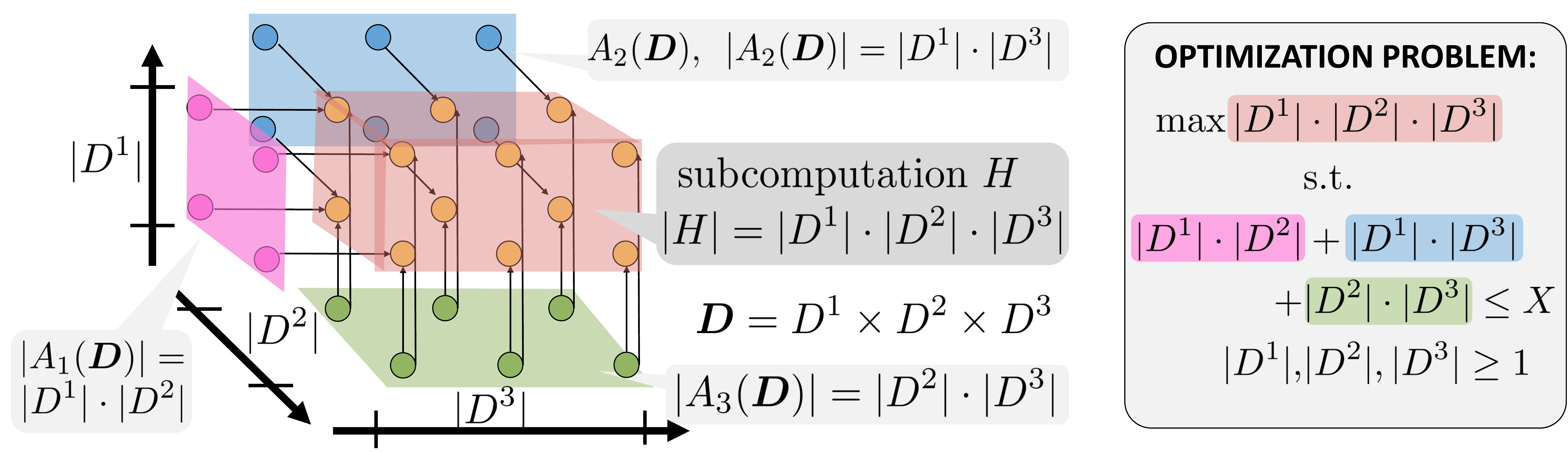}
	\caption{Lemma~\ref{lma:rectTiling} bounds the set sizes (both the 
	subcomputation's $H$ and input access sets' $|A_j(\bm{D})|$) with the 
	number of 
	values $|D^t|$ each iteration variable $\psi^t$ takes during the 
	subcomputation.}
	\label{fig:lemma3}
\end{figure}
\begin{restatable}{lma}{rectTiling}
	\label{lma:rectTiling}
	{Given the ranges of all iteration variables \mbox{$D^t, t = 
			1,\dots,l$} 
		during 
		subcomputation 
		$H$, if $|H| = \prod_{t=1}^{l}|D^t|$, then 
		$\forall j = 1, \dots, 
			m :$ $|A_j(\bm{D})| = $ $
			\prod_{k=1}^{dim(\bm{\phi}_j)}|D^k_{j}|$ and $|H|$ is 
			maximized 
		among 
		all valid subcomputations that iterate over $\bm{{D}} = 
		[D^1, 
			\dots, D^t]$.}
\end{restatable}

\noindent
\Intuition{Lemma~\ref{lma:rectTiling} states that if each iteration variable 
	$\psi^t, t = 1, \dots, l$ takes $|R_h^t|$ different values, then there are 
	at 
	most 
	$\prod_{t=1}^{l}|D^t|$ different iteration vectors $\bm{\psi}$ which 
	can 
	be 
	formed in $H$. So, intuitively, to maximize $|H|$, all combinations of 
	values $\psi^t$ should be evaluated. On the other hand, this also implies 
	maximization of all access sizes $|A_j(\bm{D})| = 
	\prod_{k=1}^{dim(\bm{\phi}_j)}|D^k_{j}|$.} 

To prove it, we now introduce two auxiliary lemmas:
\begin{lma}
	\label{lma:vi_bound}
	For statement $S$, the size $|H|$ of 
	subcomputation $H$ (number 
	of 
	vertices of $S$ computed during 
	$H$) is bounded by the sizes of the iteration 
	variables' 
	sets $R_h^t, t = 
	1, \dots, l$:
	\begin{equation}
	\label{eq:vmax_vol}
	|H| \le \prod_{t=1}^{l}|D^t|.
	\end{equation}
\end{lma}

\begin{proof}
	Inequality~\ref{eq:vmax_vol} follows from a combinatorial argument: 
	each 
	computation in $H$ is uniquely defined by its iteration vector $[\psi^1, 
	\dots, 
	\psi^l]$. As each iteration variable $\psi^t$ 
	takes $|R_h^t|$ 
	different values during $H$, we have $|R_h^1| \cdot |R_h^2| 
	\cdot 
	\dots \cdot |R_h^t| 
	= \prod_{t=1}^{l}|D^t|$ ways how to uniquely choose the iteration 
	vector in 
	$H$. 
\end{proof}

Now, given $\bm{D}$, we want to assess how many 
different vertices are 
accessed for each input array {$A_j$}.
Recall that this number is denoted as access size 
$|A_j(\bm{D})|$.

We will apply the same combinatorial reasoning to 
$A_j(\bm{D})$. For each access 
$A_j[\bm{\phi}_j(\bm{\psi})]$, each 
one of $\psi_j^k$, \linebreak$k 
= 1, 
\dots, dim(\bm{\phi}_j)$ iteration variables loops over set 
$R_{h,j}^k$ during subcomputation $H$.
We can thus bound size of $A_j(\bm{D})$ 
similarly to 
Lemma~\ref{lma:vi_bound}:

\begin{lma}
	\label{lma:projection_bound}
	The access size  $|A_j(\bm{D})|$ of subcomputation 
	$H$ 
	(the number of vertices 
	from the array $A_j$ required to compute $H$) is 
	bounded by the 
	sizes of $dim(\bm{\phi}_j)$
	iteration variables' sets $R_{h,j}^k, k = 
	1, \dots, dim(\bm{\phi}_j)$:
	
	\begin{equation}
	\label{eq:a_j_vol}
	\forall_{j = 1, \dots, m} : |A_j(\bm{D})| \le 
	\prod_{k=1}^{dim(\bm{\phi}_j)}|D^k_{j}|
	\end{equation}
	
	where {$D^k_{j}$} $\ni$ {$\psi^k_j$} is the set over which 
	iteration 
	variable {$\psi^k_j$}
	iterates 
	during $H$.
\end{lma}

\begin{proof}
	We use the same 
	combinatorial argument as in Lemma~\ref{lma:vi_bound}. Each vertex in 
	$A_j(\bm{D})$ is uniquely defined by $[\psi_j^1, 
	\dots, 
	\psi_j^{dim(\bm{\phi}_j)}]$. Knowing the number of different values each 
	$\psi_j^k$ 
	takes, 
	we bound the number of different access vectors 
	$\bm{\phi}_j(\bm{\psi}_h)$.
\end{proof}

\noindent \macb{Example}:	
\emph{Consider once more statement $S1$ from LU factorization in 
	Figure~\ref{fig:prog_rep}.
	We have {$\boldsymbol{\phi}_0$} = {[i, k]}, 
	{$\boldsymbol{\phi}_1$} = {[i, k]}, and 
	{$\boldsymbol{\phi}_2$} 
	= {[k, k]}. 
	Denote the iteration subdomain for subcomputation 
	$H$ as $\bm{D} = $\linebreak
	$\{${$[k^1, i^1]$}, $\dots$, 
	{$[k^{|H|}, 
		i^{|H|}]$} $\}$,
	where each variable {$k$} and {$i$} iterates over its 
	set
	{$k^g$}
	$\in \{\psi_{k,1}, \dots, \psi_{k,K}\} =$ 
	{$D^k$} and {$i^g$} 
	$\in \{\psi_{i,1}, \dots, \psi_{i,I}\} = $ 
	{$D^i$}, for $g 
	= 
	1, \dots, |H|$. Denote the sizes of these 
	sets as $|D^k| 
	= K$ and $|D^i| = I$, that is, during $H$, 
	variable 
	{$k$} takes $K$ different values and {$i$} takes $I$ 
	different values. For 
	{$\boldsymbol{\phi}_1$}, both iteration 
	variables 
	used are different: {k} and {i}.
	Therefore, we have (Equation~\ref{eq:a_j_vol}) 
	$|A_1(\bm{D})| \le K_h \cdot I_h$. On the other 
	hand, for 
	$\boldsymbol{\phi}_2$, the 
	iteration variable {$k$} is used twice. 
	Recall that the access dimension is the minimum number of different 
	iteration 
	variables that uniquely address it 
	(Section~\ref{sec:inputPrograms}), so its dimension is 
	$dim(${$A_2$}$) = 1$ and 
	the only iteration variable needed to uniquely determine 
	$\bm{\phi}_2$ is {$k$}.  Therefore,
	$|A_2(\bm{D})| \le K_h$.}

\macb{Dominator set.}
Input vertices $A_1, \dots, A_m$ form a dominator 
set of vertices $A_0$,
because any path from graph inputs to any vertex in $A_0$ 
must include \emph{at least} one 
vertex from $A_1, \dots, A_m$.
This is also the 
\emph{minimum} dominator set, because of the disjoint access property
(Section~\ref{sec:inputPrograms}): 
any path from graph inputs to any vertex in $A_0$ can 
include \emph{at most} one 
vertex from  $A_1, \dots, A_m$.

\vspace{1em}
\noindent
\emph{Proof of Lemma~\ref{lma:rectTiling}.}
For subcomputation $H$, we have 
$|\bigcup_{j=1}^{m}A_j(\bm{D})| \le X$ (by the 
definition 
of an \xpart). Again, by the disjoint access property, we have 
$\forall j_1 \ne j_2: A_{j_1}(\bm{\mathcal{D}}) \cap A_{j_2}(\bm{\mathcal{D}}) 
= \emptyset$.
Therefore, we also have 
$|\bigcup_{j=1}^{m}A_j(\bm{D})| = 
\sum_{j=1}^{m}|A_j(\bm{D})|$.
We now want to maximize $|H|$, that is to find 
$H_{max}$ to 
obtain computational intensity $\rho$
(Lemma~\ref{lma:compIntensityPhi}). 

Now we prove that to maximize $|H|$, 
inequalities~\ref{eq:vmax_vol} 
and~\ref{eq:a_j_vol} must be tight (become equalities). 

From proof of Lemma~\ref{lma:vi_bound} it follows that 
$|H|$ is maximized 
when iteration vector $\bm{\psi}$ takes all possible 
combinations of 
iteration variables $\psi^t \in D^t$  during $H$. But, as we visit each
combination of all $l$ iteration variables, for each access $A_j$
every combination of its $[\psi_j^1, \dots, \psi_j^{dim(\bm{\phi}_j)}]$ 
iteration 
variables is also visited.  
Therefore, for every $j = 1, \dots, m$, each access size 
$|A_j(\bm{\mathcal{D}})|$ is maximized  
(Lemma~\ref{lma:projection_bound}), 
as access functions are injective, which implies that for each combination 
of  $[\psi_j^1, 
\dots, \psi_j^{dim(\bm{\phi}_j)}]$, there is one access to $A_j$.
{{$\prod_{t=1}^{l}|R_h^t|$} is then the upper bound on 
	{$|H|$}, and its 
	tightness implies that all bounds on access sizes 
	{$|A_j(\bm{\mathcal{D}})| \le 
		\prod_{k=1}^{dim(\bm{\phi}_j)}|D^k_{j}|$} are also tight.}
\qed

\noindent
\Intuition{{\mbox{Lemma~\ref{lma:rectTiling}} states that if each 
iteration variable 
	\mbox{$\psi^t$}, \linebreak \mbox{$t = 1, \dots, l$} takes \mbox{$|D^t|$} 
	different 
	values, 
	then there are 
at 
	most \linebreak
	\mbox{$\prod_{t=1}^{l}|D^t|$} different iteration vector values 
	\mbox{$\bm{\psi}$} that 
can 
	be 
	formed in \mbox{$H$}. Thus, to maximize \mbox{$|H|$} all combinations of 
	values of \mbox{$\psi^t$} should be evaluated. On the other hand, this also 
	implies the
	maximization of all access sizes \mbox{$|A_j(\bm{D})| = 
	\prod_{k=1}^{dim(\bm{\phi}_j)}|D^k_{j}|$}. This result is more general 
	than, e.g., polyhedral \mbox{techniques~\cite{polyhedralModel,  
	general_arrays, 
	olivry2020automated}} as it does not require loop nests to be affine. 
	Instead, it solely relies on set algebra and combinatorial methods.
}}

\subsection{Finding the I/O Lower Bound}
\label{sec:iobound_singlestatement}
Denoting  $H_{max} = \argmax_{H \in 
	\mathcal{P}(X)} |H|$ as the largest subcomputation 
among all valid \mbox{$X$-partitions}, we use 
Lemma~\ref{lma:rectTiling} and combine it with the 
dominator set 
constraint from Section~\ref{sec:xpart}. Note that all access set sizes are 
strictly positive integers 
$|{D}^t| \in \mathbb{N}_+, t = 1, \dots, l$. Otherwise, if any of the 
sets is empty, no computation can be performed. However, as we 
only want to find the bound on the number of I/O operations, we relax 
the 
integer constraints and replace them with $|D^t| \ge 1$. Then,
we formulate finding $\chi(X)$ (Lemma~\ref{lma:compIntensityPhi}) as the 
following optimization problem:

\begin{align}
\label{eq:findingX}
\nonumber
\max \prod_{t=1}^{l}|{D}^t| & \hspace{2em}\text{s.t.}\\
\nonumber
\sum_{j = 1}^{m} \prod_{k=1}^{dim(\bm{\phi}_j)}|{D}^k_{j}|& \le X \\
\forall 1 \ge t \ge l : |{D}^t|& \ge 1 
\end{align} 

We then find $|H_{max}| = \chi(X)$ as a function of $X$ using 
Karush –Kuhn–Tucker (KKT) conditions~\cite{kkt}. Next, we 
solve 
\begin{equation}
\label{eq:drhodx}
\frac{d\frac{\chi(X)}{X - M}}{dX} = 0.
\end{equation}
Denoting $X_0$ as the solution to Equation~(\ref{eq:drhodx}), 
we finally obtain
\begin{equation}
\label{eq:Qlowerbound}
Q \ge |V| \frac{(X_0 - M)}{\chi(X_0)}.
\end{equation}

\noindent
\macb{Computational intensity and out-degree-one vertices.}
There exist cDAGs where every non-input vertex
has 
at least $u \ge 0$ direct predecessors that are input 
vertices 
with 
out-degree 1. We can use this fact to put an additional bound on the 
computational 
intensity.

\begin{restatable}{lma}{roOnecase}
	\label{lma:rho1case}
	If in a cDAG $G=(V,E)$ every non-input vertex has at least $u$ direct 
	predecessors with out-degree one that are graph inputs, then the maximum 
	computational intensity $\rho$ of this cDAG is bounded by $\rho \le 
	\frac{1}{u}$.
\end{restatable}

\begin{proof}
	By the definition of the red-blue pebble game, all inputs start in slow 
	memory, 
	and therefore, have to be loaded. By the assumption on the cDAG, to compute 
	any 
	non-input vertex $v \in V$, at least $u$ input vertices need to have red 
	pebbles placed on them using a load operation.
	Because these vertices do not have any other direct successors (their 
	out-degree is 1), they cannot be used to compute any other non-input vertex 
	$w$. Therefore, each computation of a non-input vertex requires at least 
	$u$ 
	unique input vertices to be loaded. 
\end{proof}

\emph{Example:} Consider Figure~\ref{fig:cdagsRho1}. In~a), each
compute vertex $C[i,j]$ has two input vertices: $A[i,j]$ with out-degree 1, and
$b[j]$ with out-degree $n$, thus $u=1$. As both 
array $A$ and vector $b$ start in the slow memory (having blue pebbles on 
each 
vertex), for each computed vertex from $C$, at least one vertex from $A$ 
has 
to be loaded, therefore $\rho \le 1$. In~b), each computation needs 
two 
out-degree 1 vertices, one from vector $a$ and one from vector $b$, 
resulting 
in $u=2$. Thus, $\rho \le \frac{1}{2}$.
\begin{figure}
	\includegraphics[width=0.9\columnwidth]{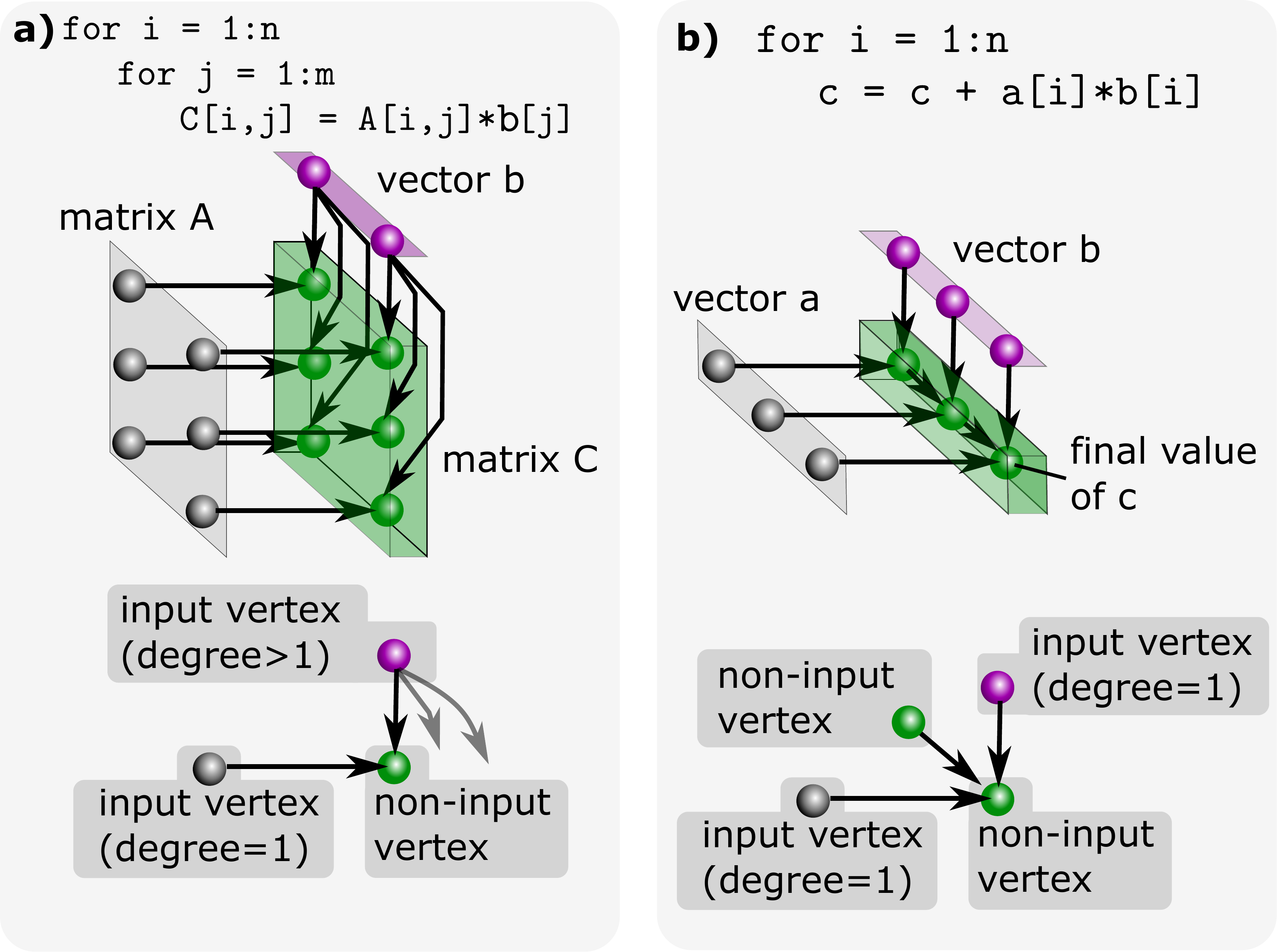}
	\caption{cDAGs with out-degree 1 input vertices. a) $u_a = 1$, $\rho_a 
		\le 
		1$. b) $u_b = 2$, $\rho_b \le \frac{1}{2}$.}
	\label{fig:cdagsRho1}
\end{figure}

\section{Data Reuse Across Multiple Statements}
\label{sec:mult_statements}

	Until now, we have analyzed each statement 
	separately. However, almost all computational kernels contain multiple 
	statements connected by data dependencies --- e.g., a column update 
	(\mbox{$S1$}) 
	and 
	a trailing matrix update (\mbox{$S2$}) in LU factorization 
	(Figure{~\ref{fig:prog_rep}}). The challenge here is that, 
	in general, I/O cost $Q$ is not composable: due to the 
	data reuse, the total I/O 
	cost of the program may be smaller than the sum of I/O 
	costs of its constituent kernels.
	In this section we examine how these dependencies influence the total I/O 
	cost 
	of a program.

We derive I/O lower bounds for programs with 
$w$ statements $S_1, \dots, S_w$ in two steps. First, we 
analyze each 
statement $S_i$ separately, as in 
Section~\ref{sec:boundsSingleStatement}. Then, we derive 
 how many loads could be avoided at most during one 
statement if another statement owned shared data. 
{There are two cases where data reuse can occur: %
	\textbf{I)} input overlap, where shared 
	arrays are inputs for multiple statements, and \textbf{II)} output 
	overlap, where the output array of one statement is the 
	input array of another. }

{\textbf{Case I)}. Assume there are \mbox{$w$} statements in the program, 
	and there are 
	\mbox{$k$} arrays \mbox{$A_j, j = 1,\dots,k$} that are shared between at 
	least 
	two 
	statements. We still evaluate each statement 
	separately, but we will subtract the upper bound on 
	shared loads $Q_{tot} \ge \sum_{i 
			= 1}^{w} Q_i - $ $\sum_{j=1}^{k} |Reuse(A_j)|$, where 
	\mbox{$|Reuse(A_j)|$} is the reuse bound on array \mbox{$A_j$}
	(Section{~\ref{sec:reuseSetSize}}).}
\textbf{Case II)}. Consider each pair of ``producer-consumer'' 
statements $S$ 
and $T$, that is, the output of $S$ is the input of 
$T$. The I/O lower bound $Q_S$ of statement $S$ does not change due to the 
reuse, as the 
same 
number of loads has to be performed to evaluate $S$. On the other hand, 
it may invalidate $Q_T$, as the  
dominator set of $T$ formulated in Section~\ref{sec:access_sizes}
may not be minimal --- inputs of a statement may not be 
graph inputs anymore. For each ``consumer'' statement $T$, we reevaluate 
$Q_T' \le Q_T$ using Lemma~\ref{lma:output_reuse}. Finally, for a 
program consisting of $w$ statements in total, connected by the output 
overlap, we have $Q_{tot} \ge 
\sum_{i = 1}^{w} Q_i'$. Note that for each ``producer'' statement $i$, 
$Q_i' = 
Q_i$, i.e. output overlap does not change their I/O lower bound.

\subsection{Case I: Input Reuse and Reuse Size}
\label{sec:reuseSetSize}

Consider two statements $S$ and $T$ that share one input array $A_i$.
Let $|A_i(\bm{R}_S)|$ denote the total number of accesses to $A_i$ during the 
I/O optimal execution of a program that contains only statement $S$. 
Naturally, $|A_i(\bm{R}_T)|$ denotes the same for a program containing only $T$.
Define $Reuse(A_i) \coloneqq \min\{|A_i(\bm{R}_S)|, |A_i(\bm{R}_T)|\}$. We then 
have: 

\begin{restatable}{lma}{inputreuse}
	\label{lma:inputreuse}
	The I/O cost of a program containing statements $S$ and 
	$T$ that share the 
	input array $A_i$ is bounded by
	$$Q_{\mathit{tot}} \ge Q_{S} + Q_{T} - Reuse(A_i),$$
	where $Q_S$, $Q_T$ are the I/O costs of a program containing only statement 
	$S$ or $T$, respectively.
\end{restatable}

\begin{proof}
	Consider an optimal sequential schedule of a cDAG $G_S$ containing 
	statement 
	$S$ only. For any subcomputation $H_s$ and its associated 
	iteration domain $\bm{R}_s$ its minimum dominator set is
	$\mathit{Dom}(H_s) = \bigcup_{j=1}^m A_j(\bm{R}_s)$. To compute $H_{S}$, 
	at least $\sum_{i=1}^{m}|A_j(\bm{R}_s)| - M$ vertices have to 
	be loaded, as only 
	$M$ vertices can be reused from previous subcomputations.
	
	We seek if any loads can be avoided in the common schedule if we add 
	statement 
	$T$, denoting its 
	cDAG $G_{S+T}$. 
	Consider a subset  
	$A_i(\bm{R}_x)$ of vertices in $A_i$. 
	
	Consider some subset of vertices in $A_i$ which potentially could be reused 
	and 
	denote it $\Theta_i$. Now 
	denote all vertices in $A_0$ (statement 
	$S$) which depend on any vertex from $\Theta_i$ as $\Theta_S$, and, 
	analogously, set $\Theta_T$ for statement $T$. Now consider these two 
	subsets $\Theta_S$ and $\Theta_T$ separately. If $\Theta_S$ is 
	computed before $\Theta_T$, then it had to load all vertices from 
	$\Theta_i$, 
	avoiding no loads compared to the schedule of $G_S$ only. Now, computation 
	of 
	$\Theta_T$ may take benefit of some vertices from $\Theta_i$, which can 
	still 
	reside 
	in fast memory, avoiding up to $|\Theta_i|$ loads. 
	
	The total number of avoided loads is bounded by the number of loads from 
	$A_i$ 
	which are shared by both $S$ and $T$. Because statement $S$ loads at most  
	$|A_i(\bm{R}_S)|$ vertices from $A_i$ during optimal schedule of $G_S$, and 
	$T$ 
	loads at most $|A_i(\bm{R}_T)|$ of them for $G_T$, the upper bound of 
	shared, 
	and possibly avoided loads is $Reuse(A_i) = \min\{|A_i(\bm{R}_S)|, 
	|A_i(\bm{R}_T)|\}$.
	
\end{proof}

The \textbf{reuse size}
is defined as $Reuse(A_i) = \min\{|A_i(\bm{R}_S)|,$  
$|A_i(\bm{R}_T)|\}$. Now, how to find $|A_i(\bm{R}_S)|$ and $|A_i(\bm{R}_T)|$?

Observe that $|A_i(\bm{R}_S)|$ is a property of $G_S$, that is, the cDAG 
containing statement $S$ only. Denote the I/O optimal schedule 
parameters of $G_S$: $V^S_{max}$, $X^S_0$, and $|A_i(\bm{R}^S_{max}(X^S_0))|$ 
(Section~\ref{sec:iobound_singlestatement}). Similarly, for $G_T$: 
$V^T_{max}$, $X^T_0$, and $|A_i(\bm{R}^T_{max}(X^T_0))|$.
We now derive: 1) at least how 
many subcomputations does the optimal schedule have: $s \ge 
\frac{|V|}{|H_{max}|}$, 2) at least how many accesses to $A_i$ are 
performed 
per optimal subcomputation $|A_i(\bm{R}_{max}(X_0))|$. Then:

\vspace{-1em}
\begin{align}
Reuse(A_i) = \min\{&|A_i(\bm{R}^S_{max}(X^S_{0}))| 
\frac{|V^S|}{|V^S_{max}|}, \\
\nonumber
|&A_i(\bm{R}^T_{max}(X^T_{0}))| \frac{|V^T|}{|V^T_{max}|}\}
\end{align}

{Note that \mbox{$Reuse(A_i)$} is an overapproximation of the actual reuse. 
Since finding the optimal schedule is 
\mbox{PSPACE-complete~\cite{redblueHard_}}, we conservatively assume that only 
the minimum number of loads from \mbox{$A_i$} is performed. Thus, 
\mbox{Lemma~\ref{lma:inputreuse}} generalizes to any number of statements $S_1, 
\dots, S_w$ sharing array  \mbox{$A_i$} --- the total number of loads from 
\mbox{$A_i$} is lower-bounded by a maximum number of loads from \mbox{$A_i$} 
among \mbox{$S_j$, $\max_{j = 1, \dots, w}|A_i(\bm{R}_{S_j})|$}.  }

\subsection{Case II: Output Reuse and Access Sizes}
\label{sec:output_reuse}

Consider the case where the \emph{output} {$A_{0}$} of 
 statement 
$S$ is 
also the \emph{input} {$B_{j}$} of statement $T$. 
 Furthermore, consider subcomputation $H$ of statement $T$ (and 
its associated iteration domain $\bm{{D}}$).
Any path from the graph 
inputs to vertices in {$B_0(\bm{{D}})$} must pass 
through 
vertices in {$B_{j}(\bm{{D}})$}. 
The following question arises:
Is there a smaller set of vertices $B_{j}'(\bm{{D}})$,
$|B_{j}'(\bm{{D}})| < |B_{j}(\bm{{D}})|$ that 
every 
path from graph inputs to {$B_{j}(\bm{{D}})$}
must pass through?

Let $\rho_S$ denote computational intensity of statement $S$. With that, we can 
state the following lemma:
\begin{restatable}{lma}{outputReuse}
	\label{lma:output_reuse}
	Any dominator set of set $B_{j}(\bm{{D}})$ must be of size at 
	least $|\mathit{Dom}(B_{j}(\bm{{D}}))| \ge 
	\frac{|B_{j}(\bm{{D}})|}{\rho_S}$.
\end{restatable}

\begin{proof}
	By Lemma~\ref{lma:compIntensity}, for one loaded vertex, 
	we may 
	compute at most $\rho_S$ vertices of $A_0$. These are 
	also vertices of $B_{j}$. Thus, to compute 
	$|B_{j}(\bm{\mathcal{D}})|$ vertices of $B_{j}$, at least 
	$\frac{|B_{j}(\bm{\mathcal{D}})|}{\rho_S}$ loads must be 
	performed. We just need to show that at least that many  
	vertices have to be in any dominator set 
	$\mathit{Dom}(B_{j}(\bm{\mathcal{D}}))$. 
	Now, consider the converse: There is a vertex set
	$D = \mathit{Dom}(B_{j}(\bm{\mathcal{D}}))$ such that $|D| < 
	\frac{|B_{j}(\bm{\mathcal{D}})|}{\rho_S}$. But that would mean, 
	that we could potentially compute all $|B_{j}(\bm{\mathcal{D}})|$ 
	vertices by only loading $|D|$ vertices,   
	violating Lemma~\ref{lma:compIntensity}.
\end{proof}

\begin{corollary}
	\label{cor:output}
	Combining Lemmas~\ref{lma:output_reuse} 
	and~\ref{lma:rectTiling},
	the data access size of $|B_{j}(\bm{{D}})|$ during 
	subcomputation $H$ is
	\begin{equation}
	\label{eq:inputOutputReuseSize}
	|\mathit{Dom}(B_{j}(\bm{{D}}))| \ge 
	\frac{\prod_{k=1}^{dim(\bm{\phi}_j)}|D_{j}^k|}{\rho_S}.
	\end{equation}
\end{corollary}

{Similarly to \textbf{case I}, this result also generalizes to multiple 
``consumer'' statements that reuse the same output array of a ``producer'' 
statement, as well as any combination of input and output reuse for multiple 
arrays and statements. Since the actual I/O optimal schedule is unknown, the 
general strategy to ensure correctness of our lower bound is to consider each 
pair of interacting statements separately as one of these two cases. Since both 
\mbox{Lemma~\ref{lma:inputreuse} and~\ref{lma:output_reuse}} overapproximate 
the reuse, the final bound may not be tight - the more inter-statement reuse, 
the more overapporixmation is needed. Still, this method can be successfully 
applied to derive \emph{tight} I/O lower bounds for many linear algebra 
kernels, such as matrix factorizations, tensor products, or solvers.}

\section{General Parallel I/O Lower Bounds}
\label{sec:parredblue}

{We now establish how our method applies to a parallel machine with $P$ 
	processors (Section{~\ref{sec:machineModel}})}.
{Since we target large-scale distributed systems, our parallel pebbling 
model differs from the one introduced e.g. by Alwen and 
Serbinenko\mbox{~\cite{parallelPebbling}}, which is inspired by shared-memory 
models 
like \mbox{PRAM~\cite{pram}}. Instead, we disallow sharing memory (pebbles) 
between the processors, and enforce explicit communication --- analogous to the 
load/store operations --- using red and blue pebbles. This allows us to better 
match the behavior of real, distributed applications that use, e.g., MPI.}

Each processor $p_i$ owns its private fast memory 
that can hold up to $M$ words, represented in the cDAG as 
$M$ vertices of color $p_i$.  Vertices with 
different colors (belonging to different processors) cannot be 
shared between these processors, but any number of different 
pebbles may be placed on one vertex.

All the standard red-blue pebble game rules apply with the 
following modifications:
\begin{enumerate}
	\item \textbf{compute.} Uf all direct predecessors of 
	vertex 
	$v$ have pebbles of $p_i$'s color placed on them, one 
	can place a pebble of color $p_i$ on $v$ (no sharing 
	of pebbles between processors),
	\item \textbf{communication.} If a vertex $v$ has \emph{any} 
	pebble 
	placed on it,
	a pebble of 
	any other color may be placed on this vertex.
\end{enumerate}

From this game definition it follows that from a perspective 
of a single processor $p_i$, any data is either local (the 
corresponding vertex has $p_i$'s pebble placed on it) or remote, 
{without a distinction on the remote location (remote access cost is 
	uniform).}

\begin{restatable}{lma}{parallelIO}
	\label{lma:parallelIO}
	The minimum number of I/O operations in a parallel 
 pebble game,
	played on a cDAG with $|V|$ vertices with $P$ processors 
	each equipped with $M$ pebbles, is $Q \ge \frac{|V|}{P \cdot \rho}$,
	where $\rho$ is the maximum computational intensity, which is independent 
	of $P$
	(Lemma~\ref{lma:compIntensity}).
\end{restatable}

\begin{proof}
	Following the analysis of Section~\ref{sec:boundsSingleStatement} and 
	the parallel machine 
	model (Section~\ref{sec:parredblue}), the computational 
	intensity $\rho$ is independent of a number of parallel 
	processors - it is solely a property of a cDAG and private 
	fast memory size $M$. Therefore, following 
	Lemma~\ref{lma:compIntensity}, what changes with $P$ is the 
	volume of computation $|V|$, as now at least one processor will 
	compute at least $|V_p| = \frac{|V|}{P}$ vertices. 
	By 
	the definition of the computational intensity, the 
	minimum number of I/O operations required to pebble these 
	$|V_p|$ vertices is $\frac{|V_p|}{\rho}$.
\end{proof}

\section{\hspace{-0.0em}{{I/O Lower Bounds of Parallel Factorization 
			Algorithms}}}
\label{sec:lu_lowerbound}
We gather all the insights from 
	Sections{~\ref{sec:background}} to{~\ref{sec:parredblue}} and use them to 
	obtain 
	the parallel I/O lower bounds of LU and Cholesky factorization algorithms, 
	which we use to develop our communication-avoiding implementations.

\macb{Memory size.}  Clearly, $M \ge N^2/P$, as otherwise the input cannot 
fit into the collective memory of all processors. Furthermore,
in the following sections, we analyze the \emph{memory-dependent}
communication cost~\cite{general_arrays}. That is, following Solomonik et al.~\cite{2.5DLU},
we assume $M \le N^2/P^{2/3}$. This is an upper bound on the amount of memory per processor
that can be efficiently utilized under the assumptions that 1) initially, the input is not 
replicated (every element of input matrix A resides in exactly one location of one of the processors);
2) every processor performs $\Theta(N^3/P)$ elementary computations. For larger $M$, the communication
cost transitions to the \emph{memory-independent} regime~\cite{general_arrays}. All presented
memory-dependent lower bounds and algorithmic costs can be easily translated to memory-independent
version by plugging in the upper bound on the size of the usable memory.

\subsection{LU Factorization}
\label{sec:lu_bound}

In our I/O lower bound analysis we omit the row pivoting, since swapping rows 
can increase the I/O cost by at most $N^2$, which is the cost of permuting the 
entire matrix. However, the total I/O cost of the LU factorization is 
$\mathcal{O}(N^3 / \sqrt{M})$~\cite{2.5DLU}.

LU factorization (without pivoting) contains two statements
(Figure~\ref{fig:prog_rep}).
Observe that we can use 
Lemma~\ref{lma:rho1case} (out-degree one vertices) for 
statement $S1: $\texttt{ A[i,k] 
= A[i,k] / A[k,k]}.
The 
loop nest depth is $l_{S1} =2$, with iteration variables $\psi^1 = $ \texttt{k} 
and 
$\psi^2 = $ \texttt{i}. The dimension of the access function vector 
\texttt{(k,k)} is 1, revealing potential for data reuse: every input vertex 
\texttt{A[k,k]} is accessed $N-k$ times and used to compute vertices 
\texttt{A[i,k]}, $k+1 <= i < N$. However, the access function vector 
\texttt{(i,k)} has 
dimension 2; therefore, every compute vertex has one 
direct predecessor with out-degree one, which is the previous version of 
element \texttt{A[i,k]}. Using Lemma~\ref{lma:rho1case}, we therefore have 
$\rho_{S1} \le 1$.

We now proceed to statement $S2: $\texttt{ A[i,j] = A[i,j]  - A[i,k] * A[k,j]}.
	Let	{$|D^k| = K$}, {$|D^i| = I$}, 
{$|D^j| = 
	J$}. 
Observe that there is  
an output reuse (Section{~\ref{sec:output_reuse}} and 
Figure~\ref{fig:prog_rep}, red arrow) of \texttt{A[i,k]} 
between 
statements {$S1$} and 
{$S2$}. We therefore have the following
access size in statement \textbf{S2:} 
{$|A_2(\bm{D}_{S2})| = |$ \texttt{A[i,k]}$| = 
	\frac{I K}{\rho_{S1}} \ge I K$}
(Equation{~\ref{eq:inputOutputReuseSize}}). Note that in this case where 
the 
computational intensity is {$\rho_{S1} \le 1$},
the output reuse does not 
change 
the 
access size {$|A_2(\bm{D}_{S2})|$} of statement {$S2$}. This 
follows the 
intuition that it is not beneficial to recompute vertices if the 
recomputation 
cost is higher than loading it from the memory. Denoting 
$H_{S2}$ as the maximal subcomputation for statement $S2$ 
over the subcomputation domain $\bm{D}$, we 
have the following (Lemma~\ref{lma:rectTiling}):
\begin{itemize}[leftmargin=0.75em]
	\item $|H_{S2}| = K I J$
	\item $|A_1(\bm{D})| = |$ \texttt{A[i,j]} $| = I J$
	\item $|A_2(\bm{D})| = |$ \texttt{A[i,k]} $| = I K$
	\item $|A_3(\bm{D})| = |$ \texttt{A[k,j]} $| = K J$
	\item $|\mathit{Dom}(H_{S2})| = |A_1(\bm{D})| + 
|A_2(\bm{D})| + 	|A_3(\bm{D})|=  IJ + IK + KJ$
\end{itemize}

We then solve the optimization problem from 
Section~\ref{sec:iobound_singlestatement}:
\begin{align}
\nonumber
\max \text{\hspace{0.5em}} K I J, & \hspace{2em}\text{s.t.}\\
\nonumber
 IJ + IK + K J& \le X \\
\nonumber
I \ge 1, \hspace{1em} J \ge 1&, \hspace{1em} K \ge 1 
\end{align}

Which gives $|H_{S2}| = \chi(X) = 
\Big(\frac{X}{3}\Big)^{\frac{3}{2}}$ for 
$K = I = J = \Big(\frac{X}{3}\Big)^{\frac{1}{2}}$.
Then, we find $X_0$ that minimizes the expression 
$\rho_{S2}(X) = 
\frac{|H_{max}|}{X-M}$ (Equation~\ref{eq:drhodx}), yielding $X_0 = 3M$. 
Plugging it into $\rho_{S2}(X)$, we conclude that the maximum computational 
intensity of $S2$ is bounded by $\rho_{S2} \le \sqrt{M}/2$.

We bounded the maximum computational intensities $\rho_{S1} $ and  $\rho_{S2}$, 
that 
is, the minimum number of I/O operations to compute vertices belonging 
to statements $S1$ and $S2$. As the last step, we find the total number of 
compute vertices for each statement: $|V_1| = 
\sum_{k=1}^{N}(N-k-1) = N(N-1)/2$, and $|V_2| = 
\sum_{k=1}^{N}\sum_{i=k+1}^{N}(N-k-1) = N(N-1)(N-2)/3$. Using 
Lemmas~\ref{lma:compIntensity} (bounding I/O cost with the computational 
intensity) and~\ref{lma:parallelIO} (I/O cost of the parallel machine), the 
parallel I/O lower bound for LU factorization is therefore

$${ Q_{P,LU} \ge \frac{2N^3 - 6N^2 + 4N}{3P\sqrt{M}} + 
	\frac{N(N-1)}{2P} = \frac{2N^3}{3P\sqrt{M}} + 
	\mathcal{O}\Big(\frac{N^2}{P}\Big)}.$$

Previously, Solomonik et al.~\cite{2.5DLU} 
established the asymptotic I/O bound for sequential execution 
$Q = \mathcal{O}(N^3/ \sqrt{M})$. Recently, Olivry et 
al.~\cite{olivry2020automated} derived a tight leading term 
constant $Q \ge 2N^3/ (3\sqrt{M})$. To the best of our 
knowledge, our result is the first non-asymptotic bound for 
parallel execution.
The generalization from the sequential to the parallel bound is 
straightforward. Note, however, that this is only the case due to our 
pebble-based execution model, and it may thus not apply to other parallel 
machine models.

\subsection{Cholesky Factorization}

We proceed analogously to our derivation of the LU I/O bound --- 
here we just briefly outline the steps. 
The algorithm contains three statements (Listing~\ref{lst:cholesky}).
For statements $S1$ and $S2$, we can again use 
Lemma~\ref{lma:rho1case} (out-degree-one vertices).
For $S1: $\texttt{ L(k,k) 
	= 
	sqrt(L(k,k))}, the loop nest depth is 
$l_1 =1$, 
we have a single iteration variable $\psi^1 = $ \texttt{k}, and a single input 
array $A_1=$ \texttt{L} with the access function $\phi_1(\bm{\psi}) = 
$\texttt{(k,k)}. Since there is only one iteration variable present in 
$\phi_1$, we have $dim(\phi_1)=1 = l_1$. Therefore, for every compute vertex 
$v$ we have one direct predecessor, which is the previous version of element 
\texttt{L(k,k)}. We conclude that $\rho_{S1} \le 1$ and $|V_{S1}| = N$.

\begin{lstfloat}
\begin{lstlisting}[caption={Cholesky 
Factorization},basicstyle=\small\ttfamily, 
label={lst:cholesky}]
    for k = 1:N  
S1:    L(k,k) = sqrt(L(k,k));               
       for i = k+1:N
S2:      L(i,k) = (L(i,k)) / L(k,k);  
         for j = k+1:i
S3:        L(i,j) = L(i,j)  - L(i,k) * L(j,k);
    end; end; end;
\end{lstlisting}
\end{lstfloat}

For statement $S2: $\texttt{ L(i,k) = (L(i,k)) / L(k,k)}, we also have  
output reuse of \texttt{L(k,k)} between statements $S2$ and $S1$.
However, as with the output reuse considered in the LU 
analysis, 
the computational intensity is $\rho_{S1} \le 1$. Therefore, it 
does not change 
the dominator set size of $S2$.
We then use the same reasoning as for the corresponding statement $S1$ in LU 
 factorization, yielding $\rho_{S2} \le 1$. 

For statement $S3$, we derive its bound similarly to $S2$ of LU, with 
$\rho_{S3} = \sqrt{M}/2$ and $|V_{S3}| = \sum_{k=1}^{N}\sum_{i=k+1}^{N}(i-k-1) 
= N(N-1)(N-2)/6$. Note that compared to 
LU, the only significant difference is the iteration domain $|V_3|$. 
Even though Cholesky has one statement more -- the diagonal element 
update \texttt{L(k,k)} -- its impact on the final I/O bound is 
negligible for large $N$.

Again, using Lemmas~\ref{lma:compIntensity} and~\ref{lma:parallelIO} we 
establish the Cholesky factorization's parallel I/O lower bound:

$$	Q_{Chol} \ge Q_1 + Q_2 + Q_3 = \frac{|V_1|}{P \rho_1} + \frac{|V_2|}{P 
	\rho_2} + \frac{|V_3|}{P \rho_3} \approx \frac{N^3}{3 P \sqrt{M}} + 
\frac{N^2}{2P} + \frac{N}{P}$$

\hltext{The derived I/O lower bound for a sequential machine (\mbox{$P = 
1$}) 
improves the previous bound \mbox{$Q_{chol} \ge 
N^3/(6\sqrt{M})$}
derived by Olivry et \mbox{al.~\cite{olivry2020automated}}.} 
Furthermore, to the best 
of our knowledge, this is the first parallel bound for this kernel.

\section{Near-I/O Optimal Parallel Matrix Factorization Algorithms}
\label{sec:conflux}

We now present our parallel LU and Cholesky factorization 
algorithms. We start with the former, more complex algorithm, i.e. LU 
factorization. Pivoting in LU poses several performance challenges. First, 
since pivots are not 
known upfront, additional communication and synchronization is required to 
determine them in each step. Second, the nondeterministic pivot distribution 
between the ranks may introduce load imbalance of computation routines. Third, 
to minimize the communication a 2.5D parallel decomposition must be used, i.e. 
parallelization along the reduction dimension.
We address all these challenges with 
	{\conflux} --- a 
	near \emph{Communication Optimal LU	factorization using 
	{\xparting}}. 

\subsection{LU Dependencies and Parallelization}
\label{sec:par_lu_no_pivot}

\begin{figure}
	\includegraphics[width=1\columnwidth]{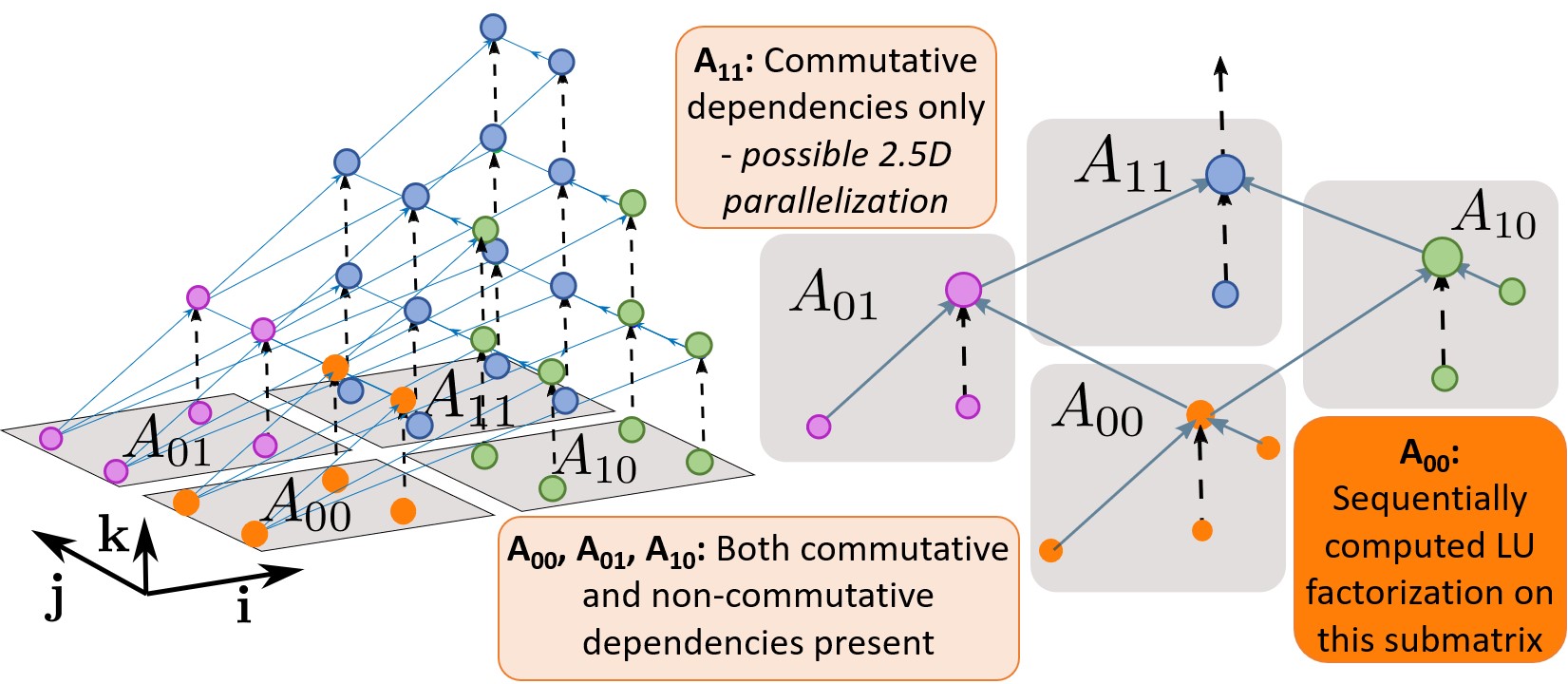}
	\caption{LU Factorization cDAG for \mbox{$n=4$} with the logical 
		decomposition into 
		\mbox{$A_{00}, A_{10}, A_{01}$}, and \mbox{$A_{11}$}. Dashed arrows 
		represent 
		commutative 
		dependencies (reduction of a value). Solid arrows 
		represent non-commutative operations, meaning that 
		any 
		parallel pebbling has to respect the induced order (e.g., no vertex 
		in 
		\mbox{$A_{11}$} can be pebbled before \mbox{$A_{00}$} is pebbled).}
	\label{fig:cdags}
\end{figure}

Due to the dependency structure of LU, the input matrix is often 
divided recursively into four submatrices $A_{00}$, $A_{10}$, 
$A_{01}$, and $A_{11}$~\cite{LUdongarra, 2.5DLU}.
{Arithmetic operations performed in LU create 
	non-commutative dependencies (Figure{~\ref{fig:cdags}}) between 
	vertices 
	in {$A_{00}$} (LU 
	factorization 
	of the top-left corner of the matrix), {$A_{10}$}, and {$A_{01}$} 
	(triangular solve 
	of left and top panels of the matrix).} Only $A_{11}$ 
	(Schur complement 
update) has no such dependencies, and may therefore be efficiently 
parallelized in the reduction dimension.
A high-level summary is presented in Algorithm~\ref{alg:conflux}. 

\begin{algorithm}
	\footnotesize
	\caption{\conflux}
	\label{alg:conflux}
	\begin{algorithmic}
			\Require \hltext{Input matrix \mbox{$A \in \mathbb{R}^{n \times n}$}}
				\Ensure \hltext{In-place factored matrix \mbox{$A$}, permutation matrix \mbox{$P$}}
		\State \hltext{\mbox{$A_1 \gets A$}} \Comment{First step}
		\State \hltext{\mbox{$P \gets I$}} \Comment{\hltext{Permutation matrix is initially identity}}
		\For{$t = 1, \dots, \frac{N}{v}$} 
		\State \textbf{1.} Reduce next block column
		\Comment{Cost: 
			$\frac{(N-t\cdot 
				v)\cdot v\cdot M}{N^2}$}
		\State \textbf{2.} $[$\emph{rows}$, P_{t+1}] \gets$ \emph{TournPivot}$(A_t, P_t)$ 
		\Comment{{\scriptsize Find next $v$ pivots.} Cost: 
			$v^2\left\lceil\log(\frac{N}{\sqrt{M}})\right 
			\rceil$}  
		\State \textbf{3.} Scatter computed $A_{00}$ and $v$ 
		pivot rows 
		\Comment{Cost: $v^2 + v$}
		\State \textbf{4.} Scatter $A_{10}$ 
		\Comment{Cost: $\frac{(N-t\cdot 
				v)v}{P}$}
		\State \textbf{5.} Reduce $v$ pivot rows 
		\Comment{Cost: 
			$\frac{(N-t\cdot 
				v)\cdot v\cdot M}{N^2}$}
		\State \textbf{6.} Scatter $A_{01}$ 
		\Comment{Cost: $\frac{(N-t\cdot 
				v)v}{P}$}
		\State \textbf{7.} \emph{Factorize}$A_{10}(A_t)$ \Comment{1D 
			parallel., 
			block-row}
		\State \textbf{8.} Send data from panel $A_{10}$ 
		\Comment{Cost: 
			$\frac{(N-t\cdot 
				v)N\cdot v}{P\sqrt{M}}$}
		\State \textbf{9.} \emph{Factorize}$A_{01}(A_t)$ \Comment{1D 
			parallel., block-column} 
		\State \textbf{10.} Send data from panel $A_{01}$ 
		\Comment{Cost: 
			$\frac{(N-t\cdot 
				v)N\cdot v}{P\sqrt{M}}$}
		\State \textbf{11.} \emph{Factorize}$A_{11}(A_t)$ 
		\Comment{2.5D parallel.} 
		\State $A_{t+1} \gets 
		A_t[$\emph{rows}$, v:$\emph{end}$]$ \Comment{Recursively process 
		remaining rows and columns}
		\EndFor		
	\end{algorithmic}
\end{algorithm}

\subsection{LU Computation Routines}

The computation is performed in $\frac{N}{v}$ steps, 
where $v$ is a tunable block size.
In each step, only submatrix $A_t$ of input matrix $A$ is updated. Initially, 
$A_t$ 
is set to $A$. \hltext{\mbox{$A_t$} can be further viewed as composed of four 
submatrices \mbox{$A_{00}$, $A_{10}$, $A_{01}$, and $A_{11}$}  
(see \mbox{Figure~\ref{fig:lu_decomp}}). These submatrices are distributed and updated by 
	routines \emph{TournPivot}, \emph{Factorize}\mbox{$A_{10}$}, \emph{Factorize}\mbox{$A_{01}$}, 
	and 
	\mbox{\emph{Factorize}$A_{11}$:}}

\begin{itemize}[leftmargin=0.75em]
	\item {{$\bm{A_{00}}$}. This {$v \times v$} submatrix contains 
		the first {$v$} 
		elements of the current {$v$} pivot rows. It is computed during 
		{\emph{TournPivot}}, and, 
		as 
		it is required to compute {$A_{10}$} and {$A_{01}$}, it is 
		redundantly copied 
		to all processors.}
	\item {{$\bm{A_{10}}$} and {$\bm{A_{01}}$}. 
		Submatrices {$A_{10}$} and {$A_{01}$} of sizes 
		{$(N-t\cdot 
			v) \times v$} and  {$v 
			\times (N-t\cdot v)$} are 
		distributed using a 1D 
		decomposition among all processors. They are updated using a triangular 
		solve. 1D decomposition 
		guarantees 
		that there are no dependencies between processors, so no communication 
		or 
		synchronization is performed during computation, as {$A_{00}$} is 
		already 
		owned 
		by every 
		processor.}
	\item 	{{$\bm{A_{11}}$} This {$(N-t\cdot v) \times (N-t\cdot 
			v)$} 
		submatrix
		is distributed using a 2.5D, block-cyclic distribution 
		(Figure{~\ref{fig:lu_decomp}}). 
		First, the updated submatrices {$A_{10}$} and {$A_{01}$} are 
		broadcast among the 
		processors. Then, {$A_{11}$} (Schur complement) is 
		updated.
		Finally, the first block column and {$v$} chosen pivot rows are 
		reduced, 
		which will 
		form {$A_{10}$} and {$A_{01}$} in the next iteration.}
\end{itemize}

\noindent \macb{Block size $\bm{v}$}. The minimum size of each block is 
the number of processor layers in the reduction dimension ($v 
\ge c = 
\frac{PM}{N^2}$). However, to secure high performance, this 
value should be 
adjusted to hardware parameters of a given machine (e.g., vector length, 
prefetch distance of a CPU, or warp size of a GPU). Throughout the analysis, we 
assume $v = a \cdot \frac{PM}{N^2}$ for some small 
constant $a$.

\begin{figure}
	\begin{center}
		\includegraphics[width=1.0\columnwidth]{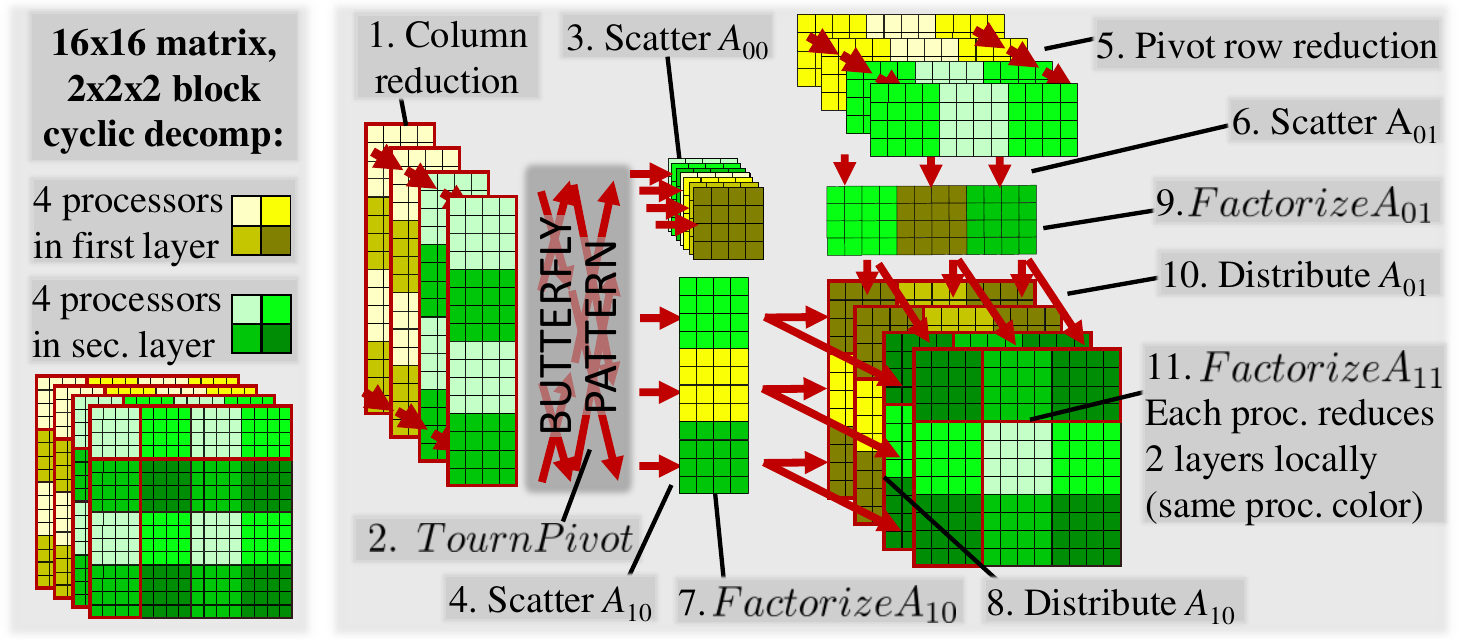}
	\end{center}
	\caption{\conflux 's parallel decomposition for $P=8$ 
		processors decomposed into a $[Px, Py, Pz] = [2, 2, 2]$ grid, 
		together with the indicated steps of   
		Algorithm~\ref{alg:conflux}. \hltext{In each iteration $t$, each 
		processor \mbox{$[pi, pj, pk]$} updates \mbox{$(2 - 
		\lfloor(t + pi)/Px \rfloor) \times (2 - \lfloor(t + 
		pj)/Py \rfloor) $} tiles of $A_{11}$. In the presented example, there 
		are \mbox{$v=4$} planes in dimension $k$ to be reduced, which are 
		distributed among $Pz=2$ processor layers (green and yellow tiles).} } 
	\label{fig:lu_decomp}
\end{figure}

\subsection{Pivoting}
\label{sec:par_lu_pivot}

Our pivoting strategy differs from state-of-the-art 
block~\cite{lapack}, tile~\cite{plasma}, or 
recursive~\cite{LUdongarra} 
pivoting
approaches in two aspects: 
\begin{itemize}[leftmargin=*]
	\item To minimize I/O cost, we do not swap pivot rows. Instead, we keep 
	track of
	which 
	rows were chosen as pivots and we use masks to update the remaining rows.
	\item To reduce latency, we take advantage of our derived block 
	decomposition 
	and use tournament pivoting~\cite{tourn_pivot}.
\end{itemize}

\begin{table*}
	\bgroup
	\small
	\def\arraystretch{1.2}
	\vspace{2em}
	\begin{tabular}{lp{2.4cm}p{1.8cm}p{3.1cm}lp{2.5cm}p{1.9cm}p{1.8cm}}
		\toprule
		& \multicolumn{3}{c}{\textbf{\conflux (LU)}} & &
		\multicolumn{3}{c}{\textbf{\chol (Cholesky)}} \\
		\cline{2-4}		\cline{6-8}
		& \textbf{description} & \textbf{comm. cost} & 
		\textbf{comp. cost} & & \textbf{description} & \textbf{comm. cost} & 
		\textbf{comp. cost} \\
		\hline 		
		\textbf{pivoting} & \emph{TournPivot} & $v^2 \lceil \log_2(\sqrt{P1}) 
		\rceil$ & 
		$v^3/3 \lceil \log_2(\sqrt{P1}) \rceil$ &  & (no pivoting) & --- & --- 
		\\
		$\bm{A_{00}}$ & local \texttt{getrf} & 0 & 0 (done 
		during 
		\emph{TournPivot}) & & \texttt{potrf} & $v^2$ & $v^3/6$ \\
		$\bm{A_{10}}$ \textbf{and} $\bm{A_{01}}$ & reduction, local 
		\texttt{trsm} & 
		$\frac{2(N-tv)vM}{N^2}$ & $\frac{2(N-tv)v^2}{2P}$ & & (similar to LU) & 
		$\frac{2(N-tv)vM}{N^2}$ & $\frac{2(N-tv)v^2}{2P}$ \\
		$\bm{A_{11}}$ & scatter, local \texttt{gemm} & $\frac{2(N-tv)v}{P}$ & 
		$\frac{(N-tv)^2 v}{P}$ & & scatter, local \texttt{gemmt} (triangular 
		\texttt{gemm}) & $\frac{2(N-tv)v}{P}$ & $\frac{(N-tv)^2 v}{2P}$ \\
		\bottomrule
		\vspace{1em}
	\end{tabular}
	\caption{Comparison of the implemented LU and Cholesky factorizations. Even 
		though Cholesky performs half as many computations (the use of 
		\texttt{gemmt} 
		instead of \texttt{gemm} in $A_{11}$), it communicates the same amount 
		of data, 
		since the number of elements needed to perform \texttt{gemm} and 
		\texttt{gemmt} 
		is the same.
	}
	\label{tab:LUvsChol}
	\egroup
\end{table*}

\noindent \macb{Tournament Pivoting.} This procedure finds $v$ 
pivot rows in each step that are then used to
mask which rows will form the new $A_{01}$ 
and then filter the non-processed rows in the next 
step. It is shown to be as stable as partial 
pivoting~\cite{tourn_pivot}, which might be an issue for, e.g., incremental 
pivoting~\cite{incrementalPivoting}.
On the other hand, it reduces the $\mathcal{O}(N)$ latency cost of partial 
pivoting, which requires step-by-step 
column reduction to find consecutive 
pivots,  to $\mathcal{O}\big(\frac{N}{v}\big)$, where $v$ is the tunable block 
size parameter.

\noindent \macb{Row Swapping vs. Row Masking}.
To achieve close to optimal I/O cost, we use a 2.5D decomposition. This, 
however, implies that in the presence of extra memory, the matrix data is 
replicated $\frac{PM}{N^2}$ times. 
This increases the row swapping cost
from  
$\mathcal{O}\big(\frac{N^2}{P})$ to 
$\mathcal{O}\big(\frac{N^3}{P\sqrt{M}}\big)$, which asymptotically matches 
the I/O lower bound of the entire factorization. Performing row swapping would 
then increase the 
constant term of the leading factor of the algorithm from 
$\frac{N^3}{P\sqrt{M}}$ to $\frac{2N^3}{P\sqrt{M}}$.
To keep the I/O cost of our algorithm as low as possible, instead of performing 
row-swapping, we only propagate pivot row indices. When the tournament pivoting 
finds the $v$ pivot rows, they are broadcast to all processors with only 
${v}$ cost per step. 

\noindent \macb{Pivoting in {\conflux}}.
In each step {$t$} of the outer loop (line 1 in 
Algorithm~\ref{alg:conflux}),
{$\frac{N}{\sqrt{M}}$} processors perform a tournament pivoting routine 
using 
a butterfly communication pattern~\cite{butterfly}. Each processor owns 
$\sqrt{M}\frac{N - vt}{N}$ rows, among which it chooses $v$ local candidate 
pivots. Then, final pivots are chosen in  {$\log(\frac{N}{\sqrt{M}})$} 
``playoff-like'' tournament rounds, after which all  
{$\frac{N}{\sqrt{M}}$} 
processors own both $v$ pivot row indices and the already factored, new 
$A_{00}$.
This result is distributed to all remaining processors (line 2).
Pivot row indices are then used to determine which processors participate in 
the reduction of the current  {$A_{01}$} (line 4).
Then, the new $A_t$ is formed by masking currently chosen rows $A_t \gets 
A_t[rows, v:end]$ (Line 12).

\subsection{I/O cost of 
		\conflux}

{We now prove the I/O cost of 
	\conflux{}, which is only a factor of $\frac{1}{3}$ higher than the 
	obtained lower 
	bound for large $N$.}

\begin{lma}
	{
		The total I/O cost of \conflux, presented in 
		Algorithm{~\ref{alg:conflux}}, is 
		{$Q_{\mathit{COnfLUX}} = \frac{N^3}{P\sqrt{M}} + 
			\mathcal{O}\left(M \right)$}.}
\end{lma}

\begin{proof}
	{
		We assume that the input matrix  {$A$} is already 
		distributed in the block cyclic layout imposed by the 
		algorithm. Otherwise, data reshuffling imposes only 
		{$\Omega\big(\frac{N^2}{P}\big)$} cost, which does not contribute 
		to 
		the 
		leading order term.
		We first derive the cost of a single iteration~{$t$} of the 
		main loop of the algorithm, proving that its 
		cost is {$Q_{step}(t) = \frac{2Nv(N-tv)}{P\sqrt{M}} + 
			\mathcal{O}\left(\frac{Mv}{N} \right)$}.
		The total cost after  {$\frac{N}{v}$} iterations is:}

	{\footnotesize
	$$Q_{\mathit{COnfLUX}} = \sum_{t=1}^{\frac{N}{v}}Q_{\mathit{step}}(t) = 
	\frac{N^3}{P\sqrt{M}} 
	+ 
	\mathcal{O}\left(M \right). $$
	}

		We define 
		{$P1 = \frac{N^2}{M}$ and $c = \frac{PM}{N^2}$}.
		{$P$} processors are decomposed into the 3D grid 
		{$[\sqrt{P1}, \sqrt{P1},c]$}. We refer to all processors 
		that share the same second and third coordinate as 
		{$[:, j, k]$}. We now examine each of 11 
		steps in Algorithm~\ref{alg:conflux}.
		
		{\noindent}
		\textbf{Step 2.}  Processors with coordinates {$[:, t \mod \sqrt{P1}, t 
		\mod c]$}  perform the tournament pivoting. Every 
		processor owns the first {$v$} elements of {$N - (t-1)v$} 
		rows, among which they choose the next  {$v$} pivots. 
		First, they locally perform the LUP decomposition to 
		choose the local {$v$} candidate rows. Then, in 
		{$\lceil \log_2(\sqrt{P1})\rceil$} rounds they exchange 
		{$v \times v$} blocks to decide on the final pivots. After the 
		exchange, these processors also hold the factorized 
		submatrix  {$A_{00}$}. 
		\emph{I/O cost per processor:}  {$v^2 \lceil 
		\log_2(\sqrt{P1})\rceil$}. 
		
		{\noindent}
		\textbf{Steps 3, 4, 6.} Factorized 
		{$A_{00}$}
		and  {$v$} pivot row indices 
		are broadcast. First  {$v$} 
		columns and {$v$} 
		pivot rows
		are scattered to all  {$P$}.
		\emph{I/O cost 
			per processor:}  {$v^2 + v + \frac{2(N-tv)v}{P}$}. 
		
		{\noindent}
		\textbf{Steps 1 and 5.} {$v$} columns and {$v$} 
		pivot rows are reduced. With high probability, pivots are evenly 
		distributed among 
		all 
		processors.
		 There are  {$c$} layers to reduce, 
		each 
		of size 
		{$(N-tv)v$}.
		\emph{I/O cost per processor:} 
		{$\frac{2(N-tv)vc}{P} = \frac{2(N-tv)vM}{N^2}$}. 
		
		{\noindent}
		\textbf{Steps 7, 9, 11.} The updates {\emph{Factorize}$A_{10}$}, 
		{\emph{Factorize}$A_{01}$}, and \linebreak {\emph{Factorize}$A_{11}$} 
		are 
		local and 
		incur no additional I/O cost.
		
		{\noindent}
		\textbf{Steps 8 and 10.} Factorized  {$A_{10}$} and {$A_{01}$} 
	are scattered among all processors.
		Each processor requires  {$\frac{v(N-tv)}{c\sqrt{P1}}$} elements 
		from  
		{$A_{10}$} and {$A_{01}$}.
		\emph{I/O cost per processor:} 
		{$\frac{2(N-tv)Nv}{P\sqrt{M}}$}. 
		
		{\noindent}
		Summing steps 1 -- 11:
		{$Q_{step}(t) = \frac{2Nv(N-tv)}{P\sqrt{M}} + 
			\mathcal{O}\left(\frac{Mv}{N} \right)$}.
\end{proof}

{Note that this cost is a factor 1/3 over the  
lower bound established in \mbox{Section~\ref{sec:lu_bound}}.
This is due to the fact that any processor can only maximally utilize its local memory
in the first iteration of the outer loop. In this first iteration, a processor updates
a total of \mbox{$\sqrt{M} \times \sqrt{M}$} elements of \mbox{$A$}. In 
subsequent iterations, however, 
the local domain shrinks as less rows and columns are updated, which leads to an underutilization
of the resources. Since the shape of the iteration space is determined by the algorithm, this
behavior is unavoidable for \mbox{$P \ge N^2/M$}. Note that the bound is 
attainable by a sequential machine, however.}

\subsection{Cholesky Factorization}

From a data flow perspective, Cholesky factorization can be viewed as a 
special case of LU factorization 
without pivoting for symmetric, positive definite matrices. Therefore, our 
Cholesky algorithm --- \chol --- heavily bases on \conflux, using the same 2.5D 
parallel decomposition, block-cyclic data distribution, and analogous 
computation routines. 

For both algorithms, the dominant cost, both in terms of computation and 
communication, 
is the $A_{11}$ update. Due to the Cholesky factorization's iteration domain, 
which exploits 
the symmetry of the input matrix, the compute cost is twice as low, as only 
one half of the matrix needs to be updated. However, the input size required to 
perform this update is the same --- therefore, the communication cost imposed 
by 
$A_{11}$ is similar.
 We list the key differences between the two 
factorization algorithms in Table~\ref{tab:LUvsChol}.

\section{Implementation}

Our algorithms are implemented in C++, using MPI for inter-node 
communication. For static communication patterns (e.g., column 
reductions) we use dedicated, asynchronous MPI collectives.
For 
runtime-dependent 
communication (e.g., pivot index distribution) we use MPI 
one-sided~\cite{mpi3-rma-overview}. For intra-node tasks, we use OpenMP and 
local BLAS calls (provided by Intel MKL~\cite{mkl}) for 
computations. Our code is available as an open-source git 
repository\footnote{https://github.com/eth-cscs/conflux}.

\noindent
\macb{Parallel decomposition. }
Our experiments show that the parallelization in the reduction dimension, while 
reducing communication volume, does incur performance overheads. This is mainly 
due to the increased communication latency, as well as smaller buffer sizes 
used for local BLAS calls. Since formal modeling of the tradeoff between 
communication volume and performance is outside of the scope of the paper, we 
keep the depth of 
parallelization in the third dimension as a tunable parameter, while providing 
heuristics-based default values.

\noindent \macb{Data distribution. }
\conflux and \chol provide ScaLAPACK wrappers by using the highly-optimized 
COSTA 
algorithm~\cite{costa2021} to transform the matrices between different layouts. 
In addition, 
they support the COSTA API for matrix descriptors, which is more general 
than ScaLAPACK's layout, as it supports matrices distributed in arbitrary 
grid-like 
layouts, processor assignments, and local blocks orderings.

\begin{table*}
	\setlength{\tabcolsep}{2pt}
	\renewcommand{\arraystretch}{0.7}
	\centering
	\footnotesize
	\vspace{2em}
	\begin{tabular}{llllll}
		\toprule
		& \textbf{MKL~\cite{mkl}}
		& \textbf{SLATE~\cite{slate}} 			
		& \textbf{CANDMC~\cite{candmc}}			
		& \textbf{CAPITAL~\cite{choleskyQRnew}} 			
		& \textbf{\conflux} / \textbf{\chol} (this work)\\
		\midrule
		\textbf{Decomposition}
		& 
		2D, panel decomp.
		&
		2D, block decomp.
		&
		Nested 2.5D, block decomp.
		& 
		2.5D, block decomp.
		& 
		1D / 2.5D, block decomp.
		\\
		\begin{tabular}{l}
			\hspace{-0.6em}
			\textbf{Block size}
		\end{tabular}
		&
		\begin{tabular}{ll}
			\begin{tabular}{l}
				\hspace{-0.5em}
				\includegraphics[width=0.047 \textwidth]
				{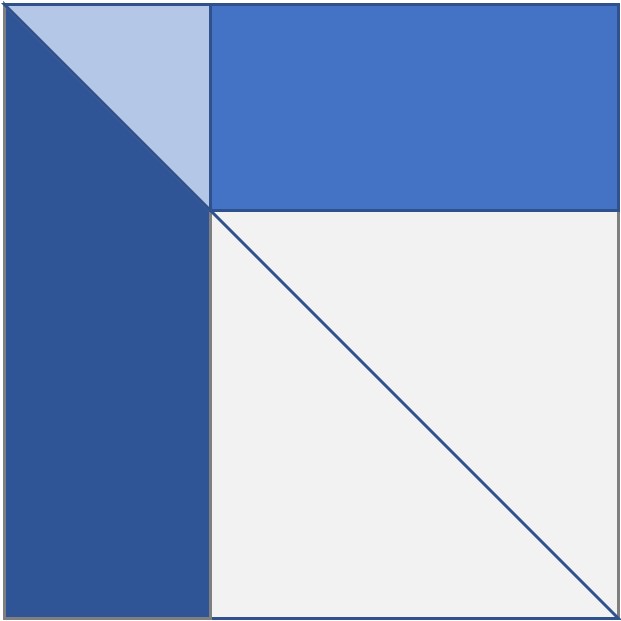}
			\end{tabular}
			&
			\begin{tabular}{l}
				user-specified
			\end{tabular}
		\end{tabular}
		&
		\begin{tabular}{ll}
			\begin{tabular}{l}
				\hspace{-0.5em}
				\includegraphics[width=0.047 \textwidth]
				{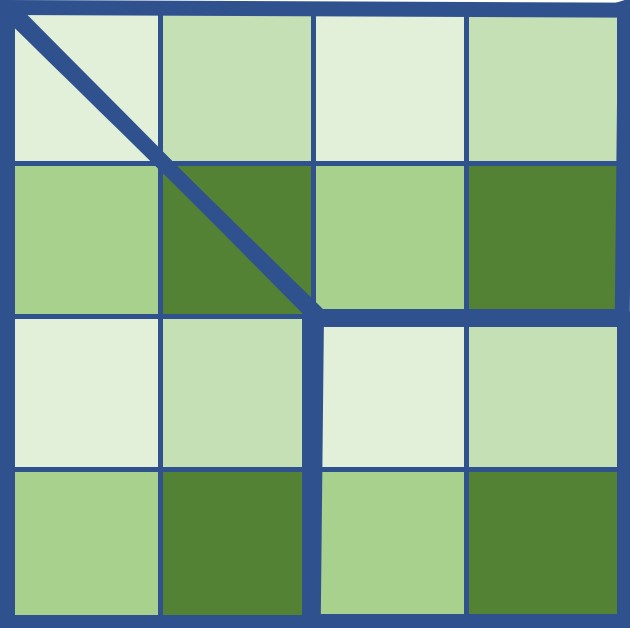}
			\end{tabular}
			&
			\begin{tabular}{l}
				user-specified,\\
				(default 16)
			\end{tabular}
		\end{tabular}
		&
		\begin{tabular}{ll}
			\begin{tabular}{l}
				\hspace{-0.5em}
				\includegraphics[width=0.047 \textwidth]
				{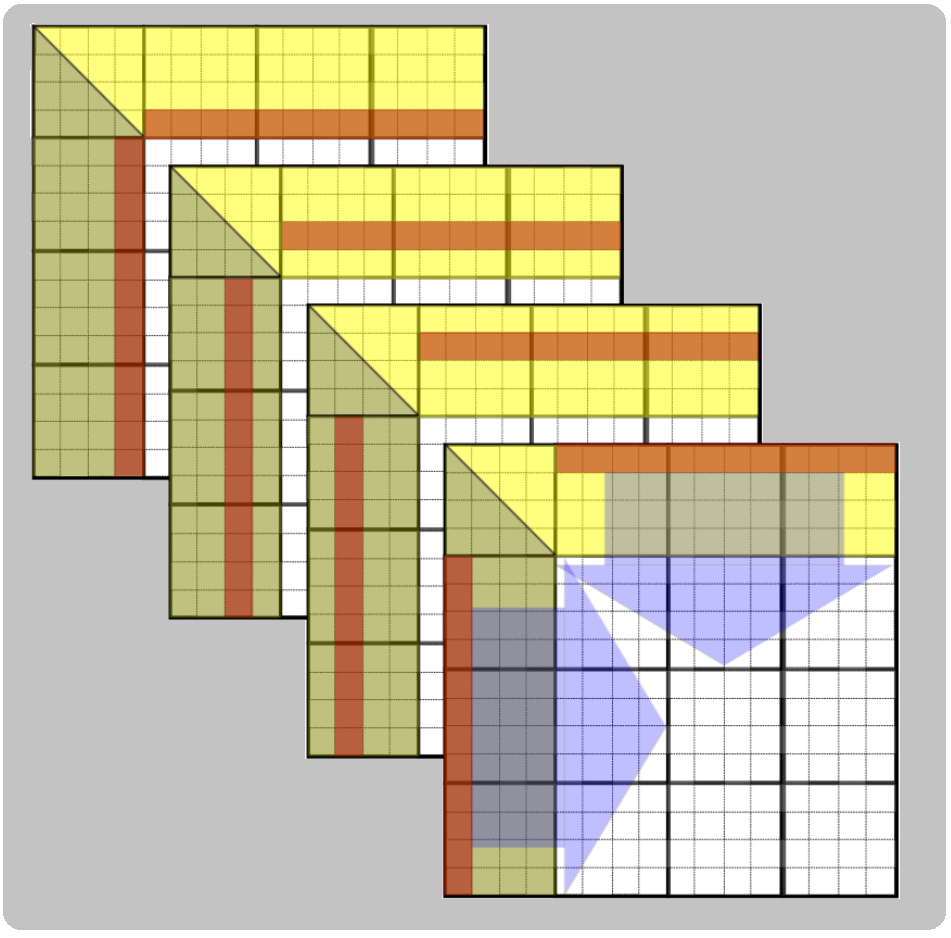}
			\end{tabular}
			&
			\begin{tabular}{l}
				$\frac{N^3}{P\cdot M}$ , 
				$\frac{N^2}{P\sqrt{M}}$
			\end{tabular}
		\end{tabular}
		&
		\begin{tabular}{ll}
			\begin{tabular}{l}
				\hspace{-0.5em}
				\includegraphics[width=0.047 \textwidth]
				{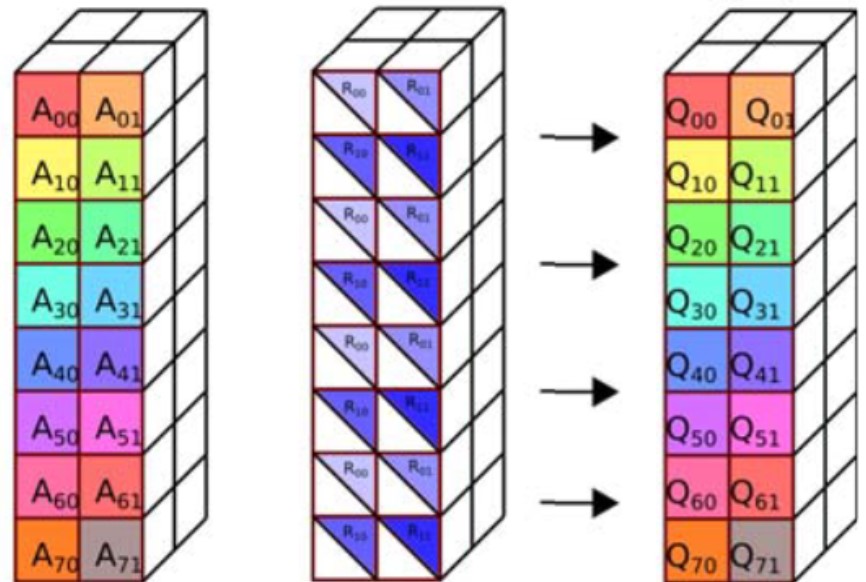}
			\end{tabular}
			&
			\begin{tabular}{l}
				user-specified
			\end{tabular}
		\end{tabular}
		&
		\begin{tabular}{ll}
			\begin{tabular}{l} 
				\hspace{-0.5em}
				\includegraphics[width=0.052 \textwidth]
				{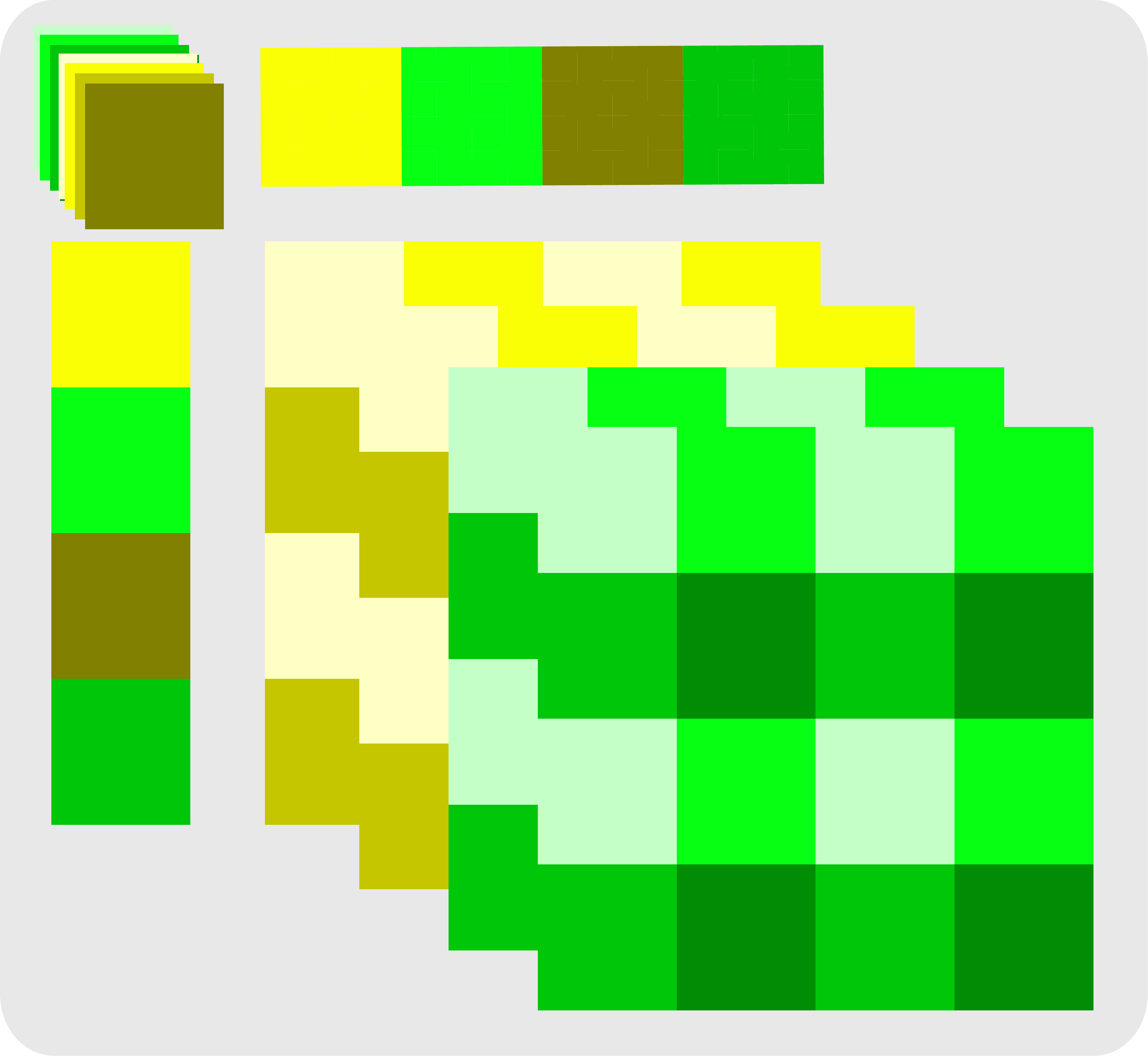}
			\end{tabular}
			&
			\begin{tabular}{l}
				optimized, $\ge \frac{P\cdot M}{N^2}$
			\end{tabular}
		\end{tabular}
		\\
		\textbf{Program parameters} & 
		required from user~\faThumbsDown & 
		required from user~\faThumbsDown & 
		provided defaults~\faThumbsOUp & 
		optimized defaults~\faThumbsOUp~\faThumbsOUp & 
		optimized defaults~\faThumbsOUp~\faThumbsOUp	
		\vspace{1em}\\
		\textbf{Parallel I/O cost} 
		& 
		$\frac{N^2}{\sqrt{P}} + 
		\mathcal{O}\Big(\frac{N^2}{P}\Big)$
		&
		$\frac{N^2}{\sqrt{P}} + 
		\mathcal{O}\Big(\frac{N^2}{P}\Big)$
		&
		$\frac{5N^3}{P\sqrt{M}} + 
		\mathcal{O}\Big(\frac{N^2}{P\sqrt{M}}\Big)$~\cite{2.5DLU}
		&
		$\frac{45N^3}{8 P\sqrt{M}} + 
		\mathcal{O}\Big(\frac{N^2}{P\sqrt{M}}\Big)$~\cite{choleskyQRnew}
		&
		$\frac{N^3}{P\sqrt{M}} + 
		\mathcal{O}\Big(\frac{N^2}{P\sqrt{M}}\Big)$
		\\
		\bottomrule
		\vspace{1em}
	\end{tabular}
	\caption{
		{Parallelization strategies and I/O cost models of 
			the 
			considered matrix 
			factorization implementations. \hltext{MKL and 
			SLATE require a user to specify the processor 
			decomposition and the block size. CANDMC provides 
			default values, but our experiments show that the 
			performance was significantly improved when we 
			tuned the parameters. \mbox{\conflux and \chol} 
			outperform all state-of-the-art libraries with 
			out-of-the-box parameters. We validated our 
			parallel I/O cost models: for MKL, 
			SLATE,\mbox{\conflux, and \chol}, the error was 
			within +/-3\%. For CANDMC and CAPITAL, we used 
			the models derived by the 
			authors\mbox{~\cite{2.5DLU, 
						choleskyQRnew}}, which 
						overappoximated the measured values 
						by approx. 30-40\%.} 
		}
	}
	
	\label{tab:comparison}
\end{table*}

\begin{figure*}[t]
	\centering	
	\subfloat[Communication volume per node for varying node 
	counts $P$ and a 
	fixed  $N=$16,384. Only the leading factors of the models 
	are shown. The models are scaled by the element size (8 
	bytes).]
	{\includegraphics[width=0.30 
	\textwidth]{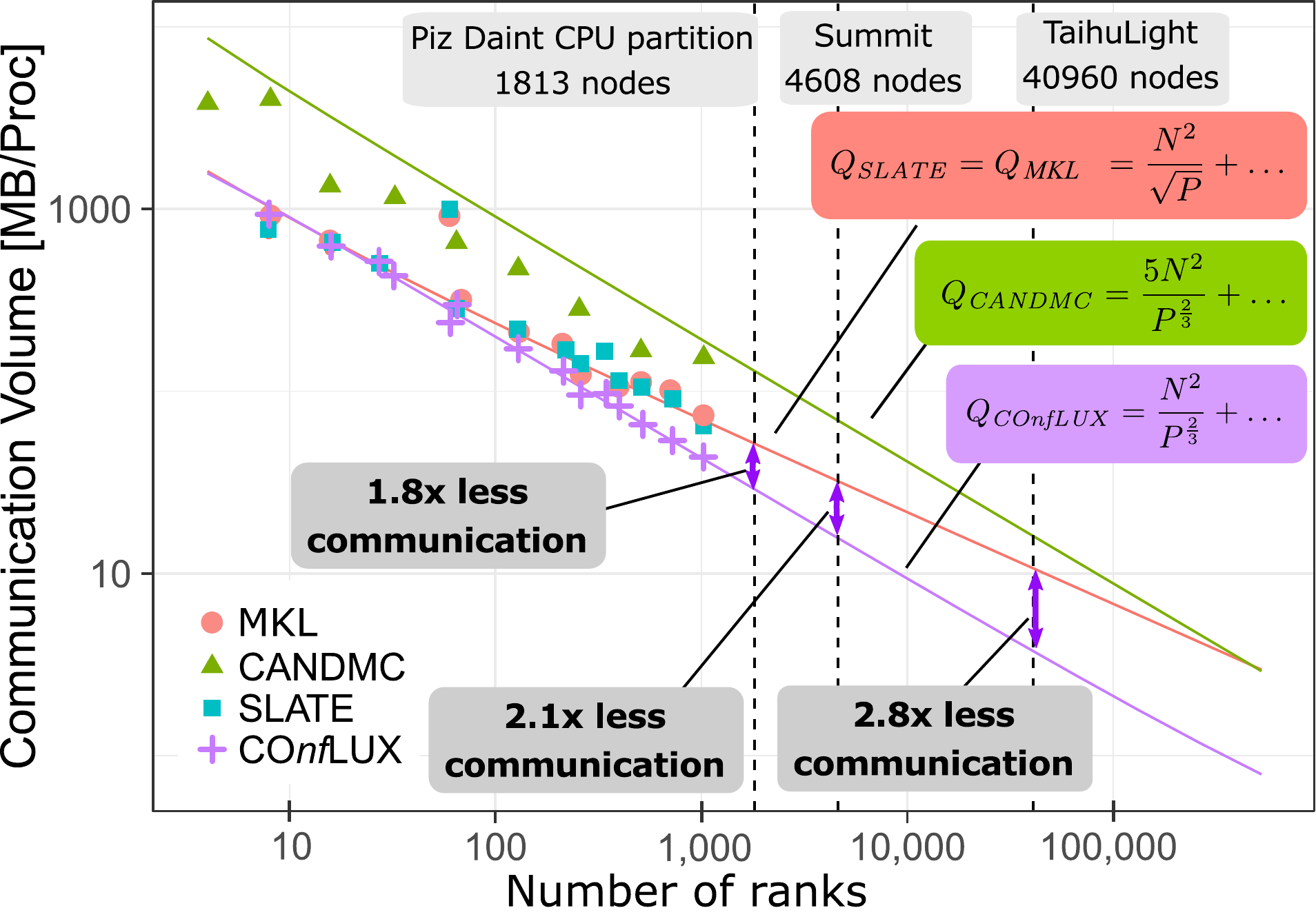}
		\label{fig:square_strong}}
	\hspace{1em}	
	\subfloat[Communication volume per node for weak scaling 
	(constant work per 
	node), 
	{\mbox{$N=3200 \cdot \sqrt[3]{P}$}}. 2.5D algorithms 
	(CANDMC and \conflux) retain constant 
	communication volume per processor.]
	{\includegraphics[width=0.31
		\textwidth]{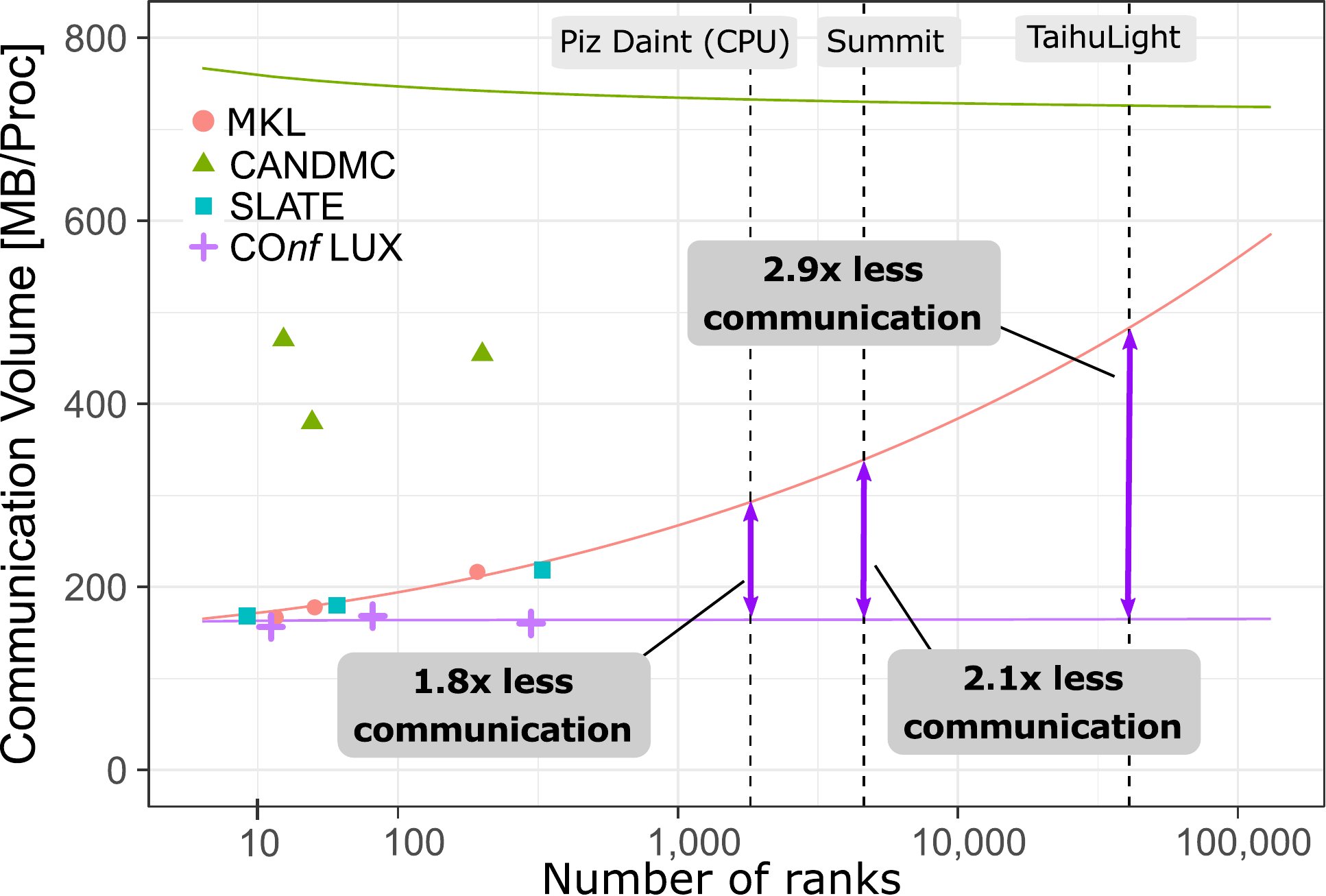}	
		\label{fig:square_weakp1}}
	\hfill
	\subfloat[Communication reduction vs. second-best 
	algorithm (M=MKL, 
	S=SLATE), for varying $P$, $N$, for both measured and 
	predicted 
	scenarios.]
	{\includegraphics[width=0.34
		\textwidth]{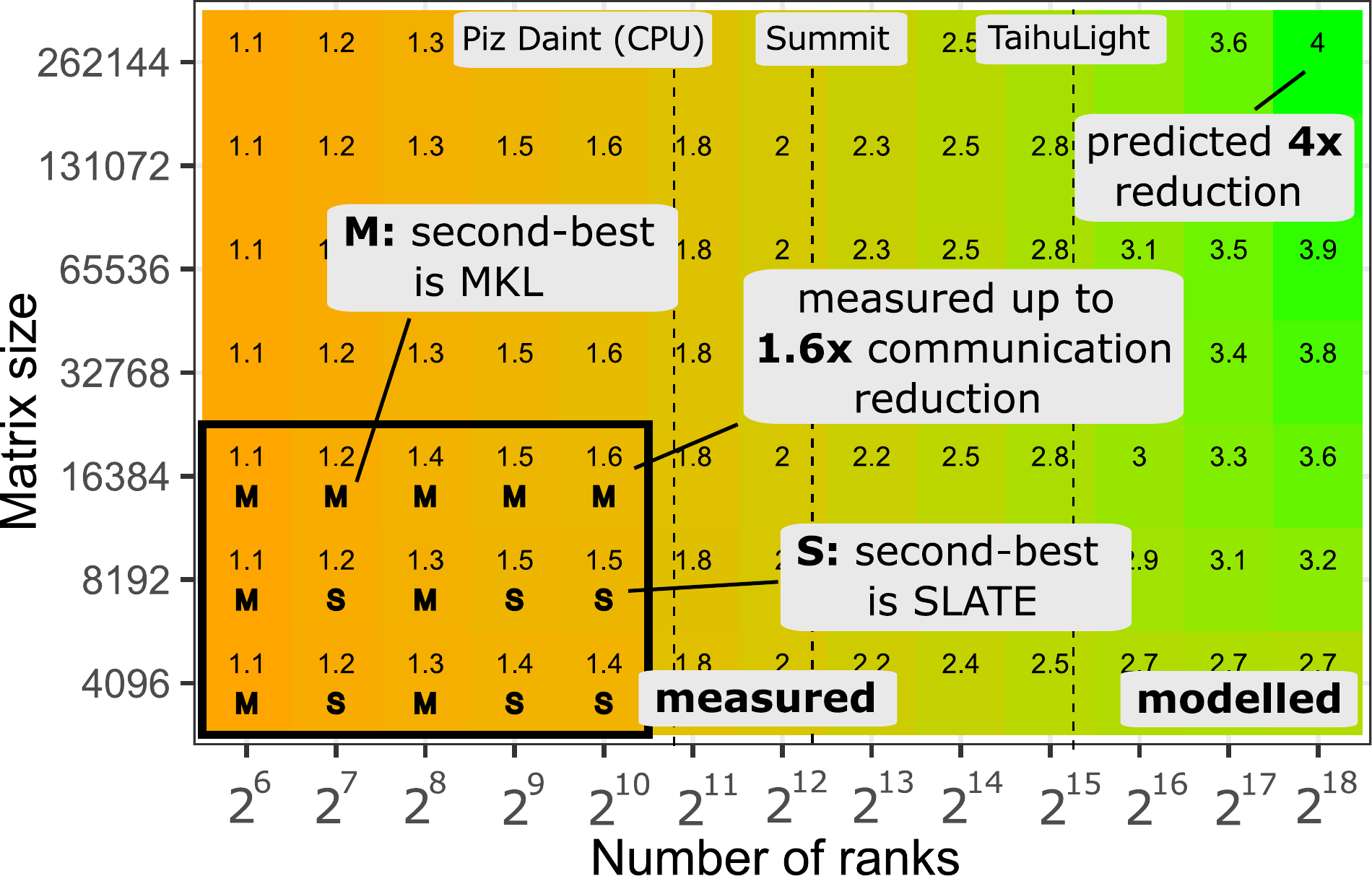}	
		\label{fig:heatmap}}
	\caption{
		{ Communication volume measurements across 
			different scenarios 
			for MKL, SLATE, CANDMC, and \conflux. In all 
			considered 
			scenarios, enough memory $M\ge {N^2}/P^{2/3}$ was 
			present 
			to allow for the maximum number of replications 
			$c = P^{1/3}$.}
	}
	\label{fig:commVolPlots}	
\end{figure*}

\section{Experimental Evaluation}
\label{sec:evaluation}

We compare \conflux and \chol with 
state-of-the-art implementations of corresponding
distributed matrix
factorizations.

\noindent
\macb{Measured values.} We measure both the I/O cost and total 
time-to-solution. For I/O, the 
aggregate communication volume in distributed runs is counted using the Score-P 
profiler~\cite{score-p}. We provide both measured 
values and theoretical cost models. Local \linebreak \texttt{std::chrono} calls 
are used for
time measurements and the maximum execution time among all ranks is reported.

\noindent
\macb{Infrastructure and Measurement.} 
We run our experiments on the XC40 partition of the CSCS Piz 
Daint supercomputer which 
comprises 1,813 CPU nodes equipped with Intel Xeon E5-2695 v4 
processors (2x18 
cores, 64 GiB DDR3 RAM), interconnected by the Cray Aries network with a 
Dragonfly network topology. Since the CPUs are dual-socket, two MPI ranks are
allocated per compute node.

\noindent
\macb{Comparison Targets.} We use 1) Intel 
MKL (v19.1.1.217). While 
the library is proprietary, our measurements 
reaffirm that, like ScaLAPACK, the implementation uses the suboptimal 2D 
processor decomposition; 2) SLATE~\cite{slate} --- a 
state-of-the-art distributed linear algebra framework 
targeted at exascale supercomputers; 3) the latest 
version of the CANDMC and CAPITAL 
libraries~\cite{candmccode, choleskycode}, which use an asymptotically-optimal 
2.5D
decomposition. The implementations and their characteristics are listed in 
Table 
\ref{tab:comparison}.

\noindent
\macb{Problem Sizes.} 
We evaluate the algorithms starting from 2 compute nodes (4 
MPI ranks) up to 512 nodes (1,024 ranks).
For each node count, matrix sizes range 
from $N=$ 2,048 to $N=2^{19}=$ 524,288, provided they fit into the 
allocated memory (e.g., LU or Cholesky factorization on a double-precision input matrix 
of dimension 262,144 
$\times$ 262,144 cannot be run on less than 32 nodes). Runs in which none of 
the libraries achieved more 
than 3\% of the hardware peak are discarded since by adding more nodes the 
performance 
starts to 
deteriorate.

Our benchmarks reflect real-world problems in scientific 
computing. 
{The High-Performance Linpack 
	benchmark uses a maximal size of 
	\mbox{$N=$ 16,473,600~\cite{top500}}. In quantum 
	physics, matrix size scales with $2^{\text{qubits}}$.
	In physical 
	chemistry or density 
	functional theory (DFT), simulations require factorizing 
	matrices of atom 
	interactions, yielding sizes ranging from \mbox{$N = $ 
	1,024} 
	up to \mbox{$N = $ 131,072~\cite{gb19, joost}}. In 
	machine 
	learning, matrix factorizations are used for inverting 
	Kronecker factors~\mbox{\cite{osawa2019large}} whose sizes are 
	usually around \mbox{$N=$ 4,096}. This motivates us to 
	focus not only on exascale problems, but also improve 
	performance for relatively small matrices (\mbox{$N \le 
	$100,000}).}

\noindent \macb{Communication Models.}
Together with empirical measurements, we put significant effort into 
understanding the underlying communication patterns of the compared LU 
factorization implementations. {Both MKL and SLATE base on the standard 
	partial pivoting algorithm using the 2D decomposition{~\cite{scalapack}}. 
	For 
	CANDMC and CAPITAL, the 
	models provided by the authors~\cite{2.5DLU,choleskyQRnew} are used. For 
	\conflux and \chol, we use the results from Section{~\ref{sec:conflux}}. 
	These 
	models 
	are summarized in Table{~\ref{tab:comparison}}. }

	\begin{figure*}[t]
		\vspace{2em}
		\centering
		\subfloat[Strong scaling, $N= 2^{17}=$ 131,072]
		{\includegraphics[width=\fw \textwidth]
			{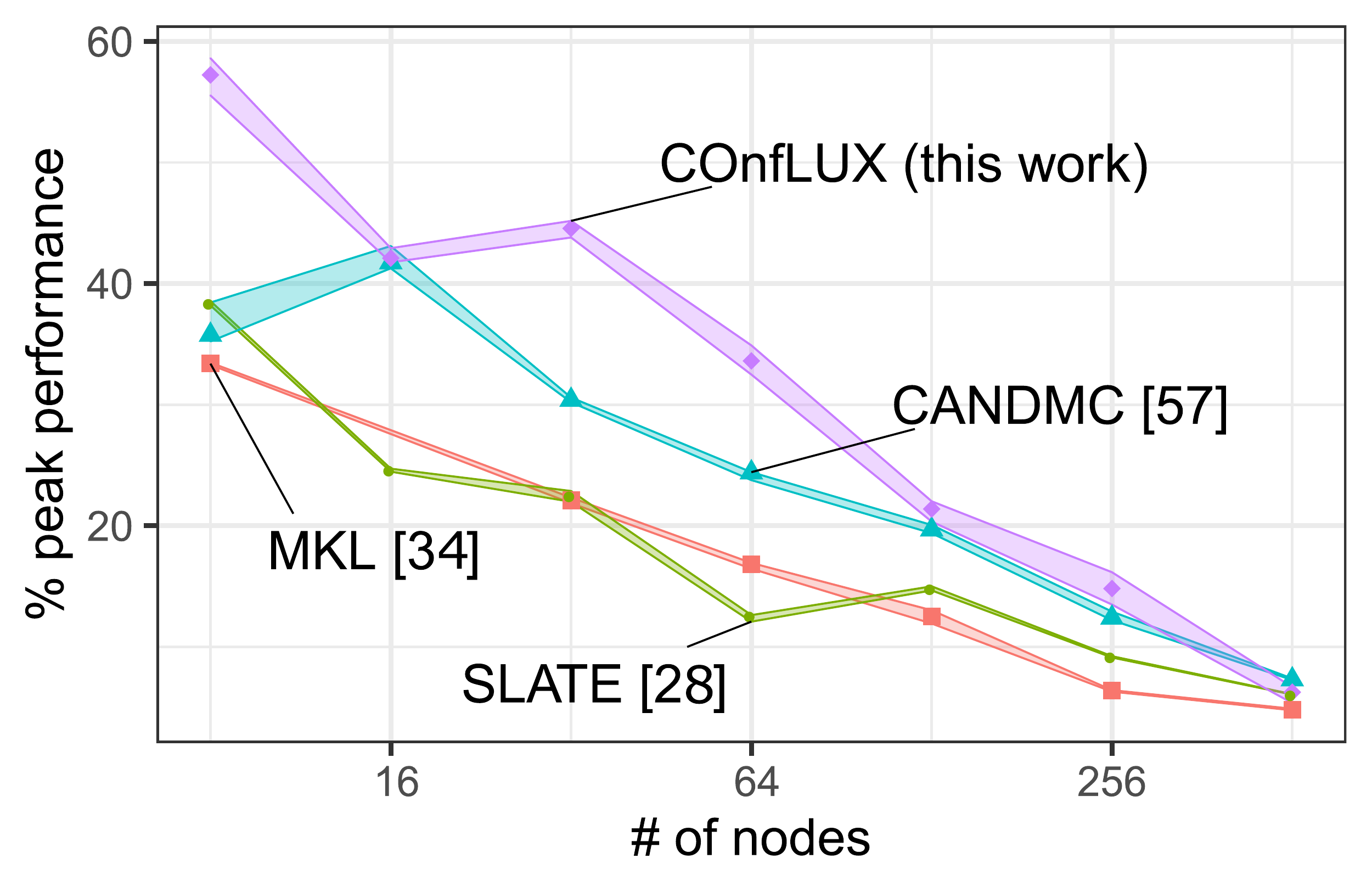}
			\label{fig:square_strong_flops}}
		\hfill
		\subfloat[Strong scaling, $N= 2^{14} = $ 
		16,384]{\includegraphics[width=\fw 
			\textwidth]			
			{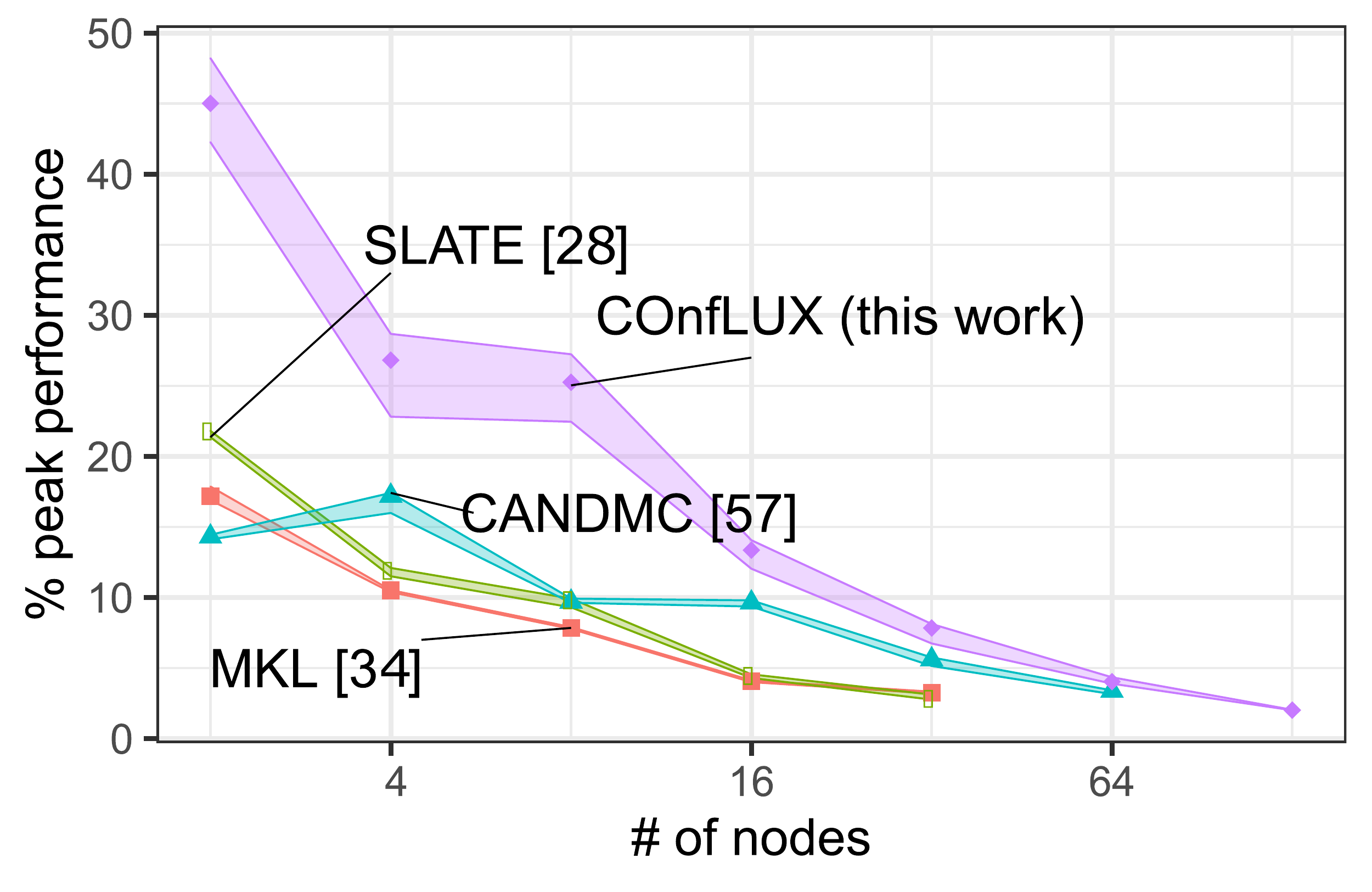}
			\label{fig:lu_strong_n16k_flops}}
		%
		\hfill
		\subfloat[Weak scaling, $N = $ 8,192$\cdot 
		\sqrt{P}$]{\includegraphics[width=\fw 
			\textwidth]			
			{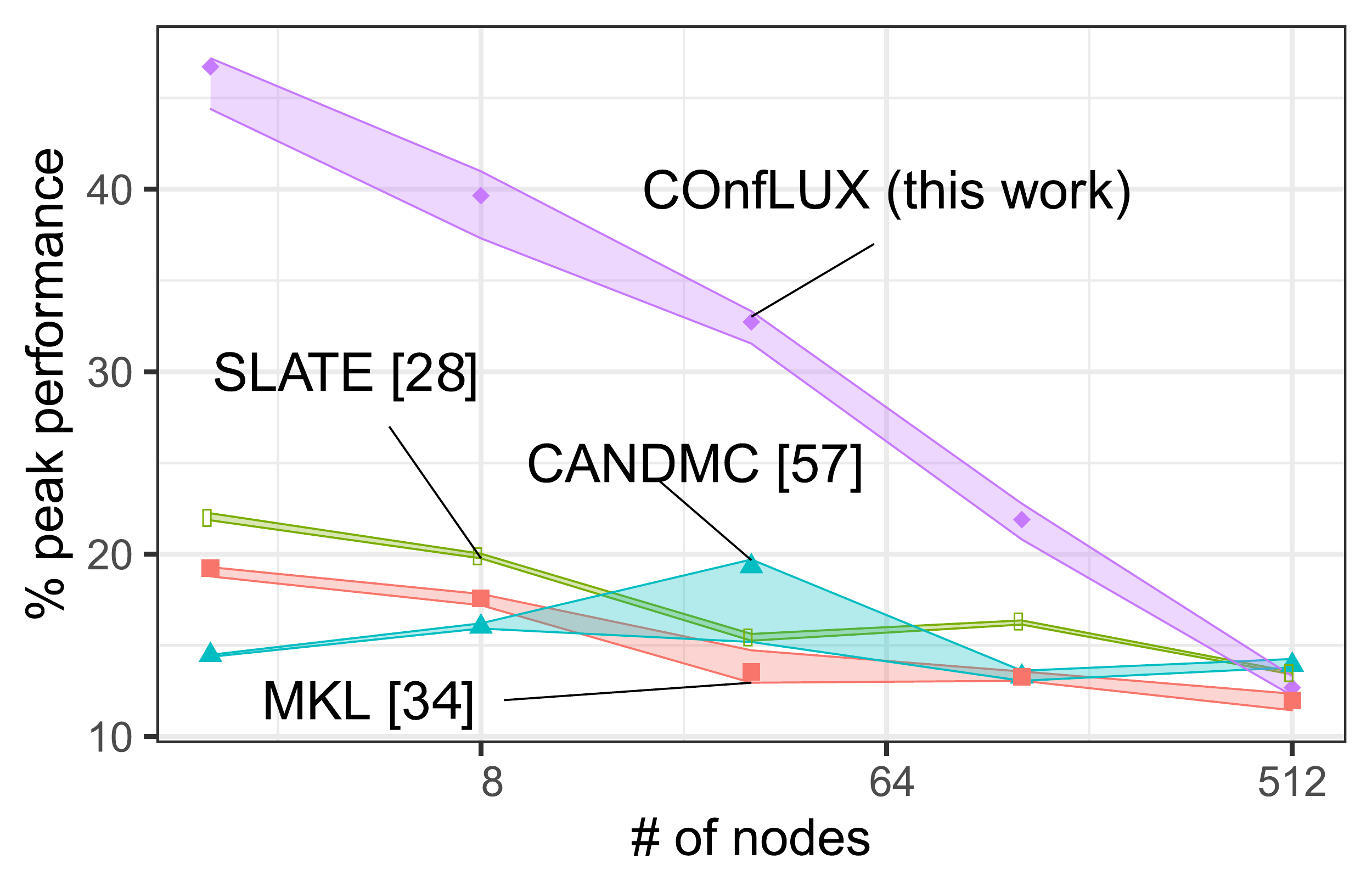}
			\label{fig:square_weakp2}}
		\caption{
			{{Achieved \% of peak performance for LU 
					factorization. We show 
					median 
					and 95\% confidence intervals.} }
		}
		%
		\label{fig:performancePlotsLU}
	\end{figure*}

\begin{figure*}[t]
		\vspace{2em}
	\centering
	\subfloat[Strong scaling, $N= 2^{17}=$ 
	131,072]{\includegraphics[width=\fw 
		\textwidth]	
		{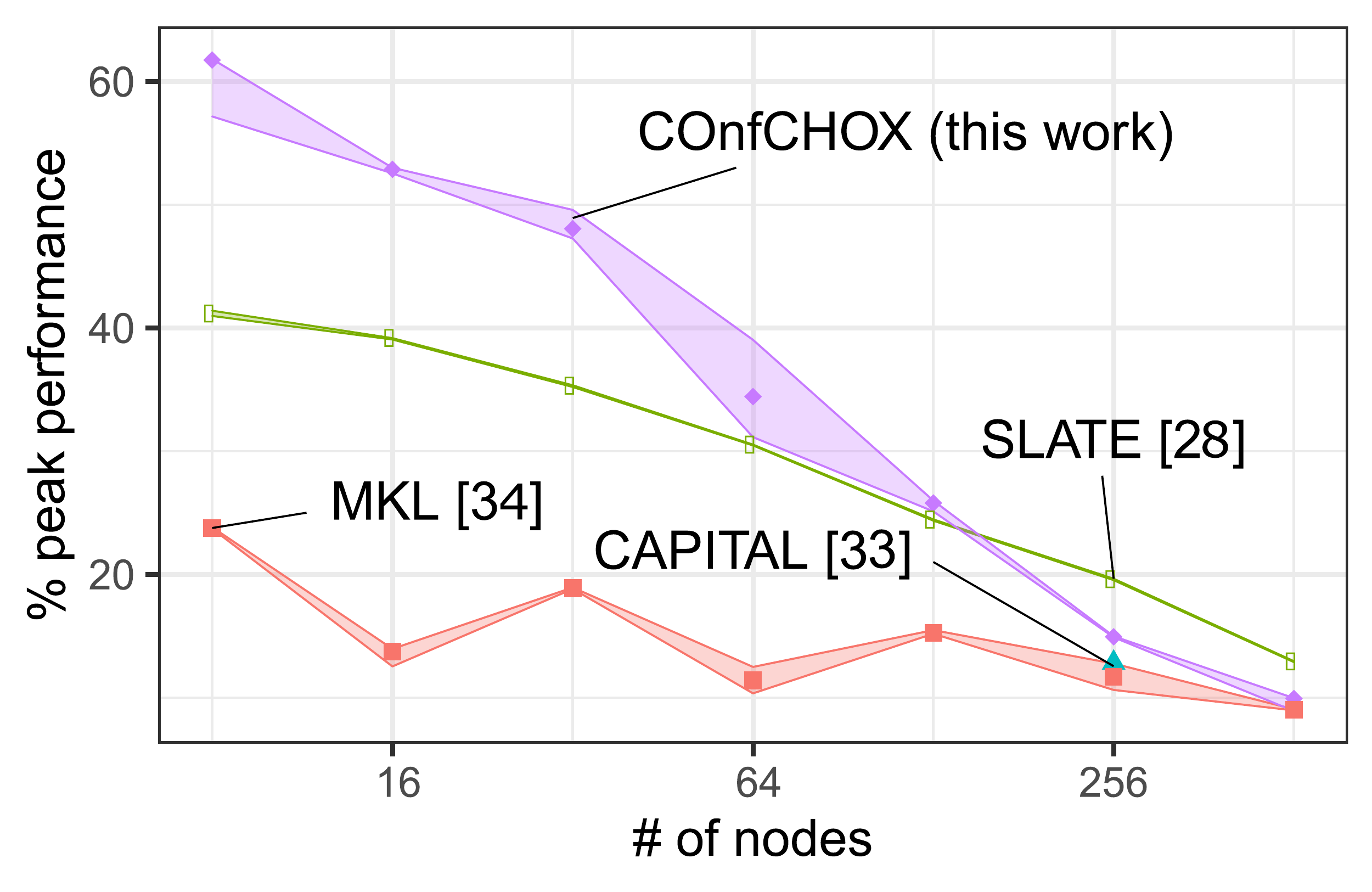}\label{fig:tall_strong}}
	\hfill
	\subfloat[Strong scaling, $N= 2^{14} = $ 
	16,384]{\includegraphics[width=\fw 
		\textwidth]	
		{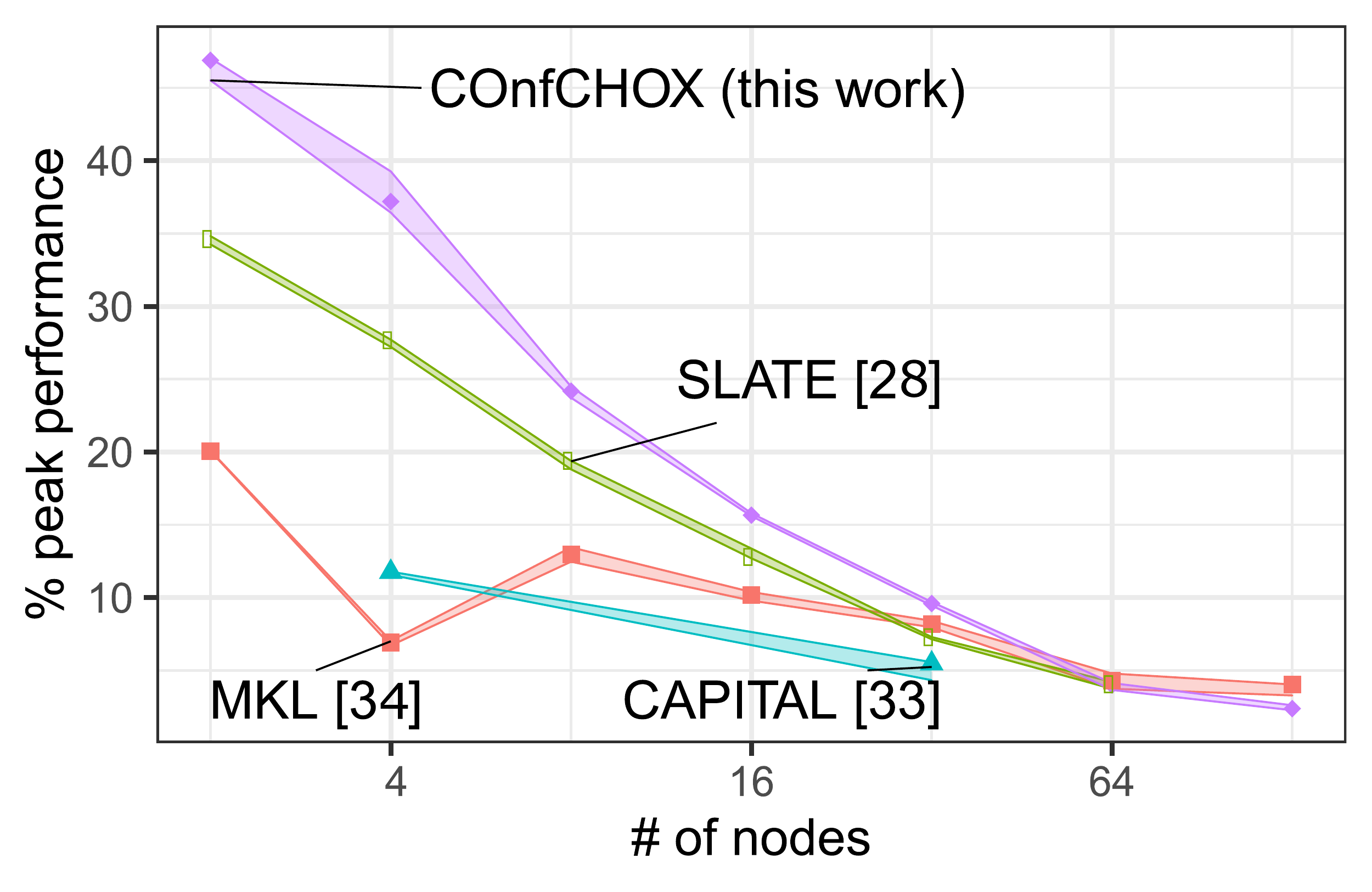}\label{fig:tall_weakp1}}
	%
	\hfill
	\subfloat[Weak scaling, $N = $ 8,192$\cdot 
	\sqrt{P}$]{\includegraphics[width=\fw \textwidth]	
		{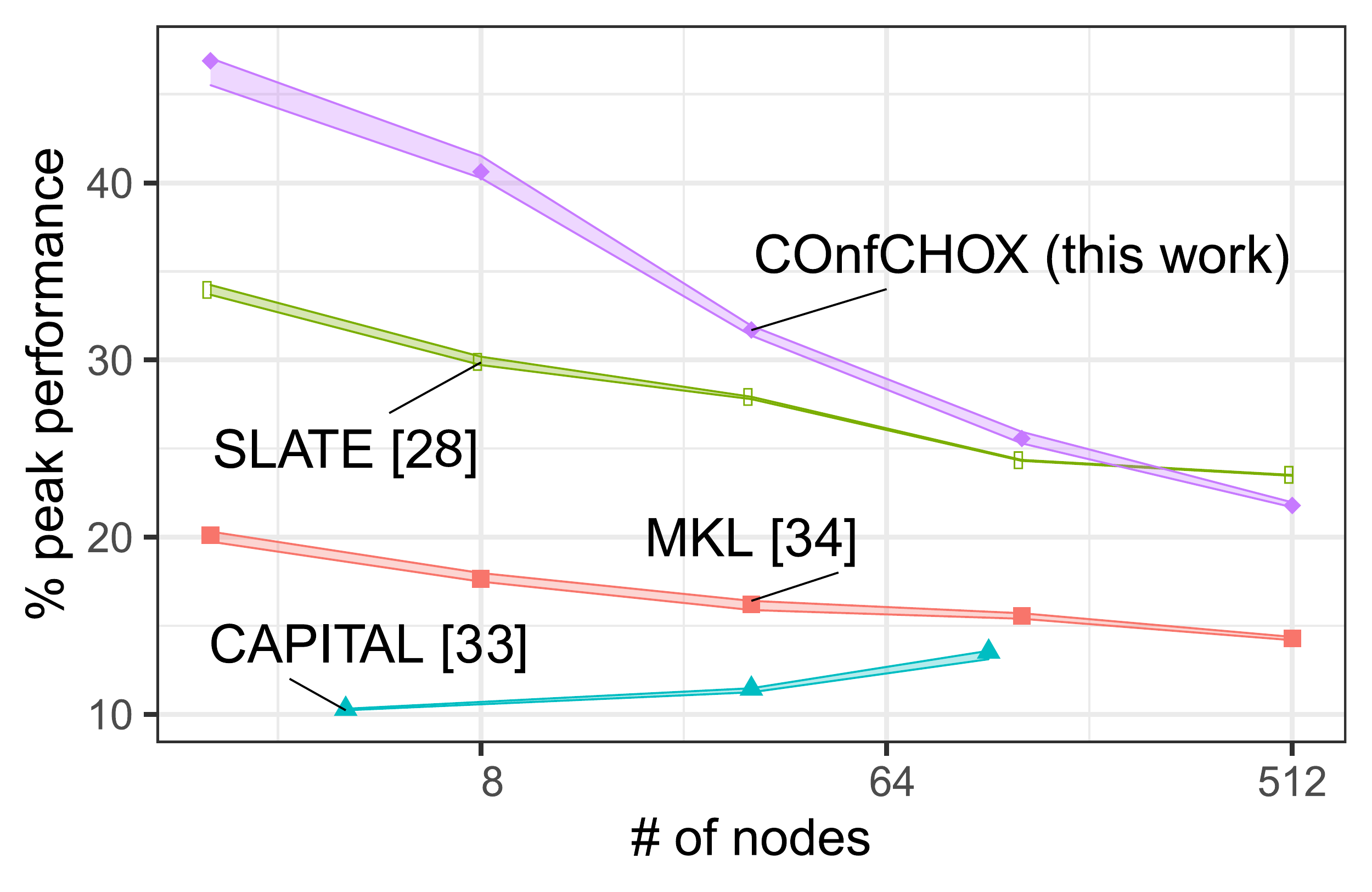}\label{fig:tall_weakp2}}
	%
	\caption{
		{{Achieved \% of peak performance for Cholesky 
				factorization. We 
				show 
				median 
				and 95\% confidence intervals.} }
	}
	\label{fig:performancePlotsCholesky}
	
\end{figure*}

\section{Results}
\label{sec:results}

	\begin{figure}
	\centering
	\subfloat
	{\hspace{-1.6em}
		\includegraphics[width=0.283 \textwidth]
		{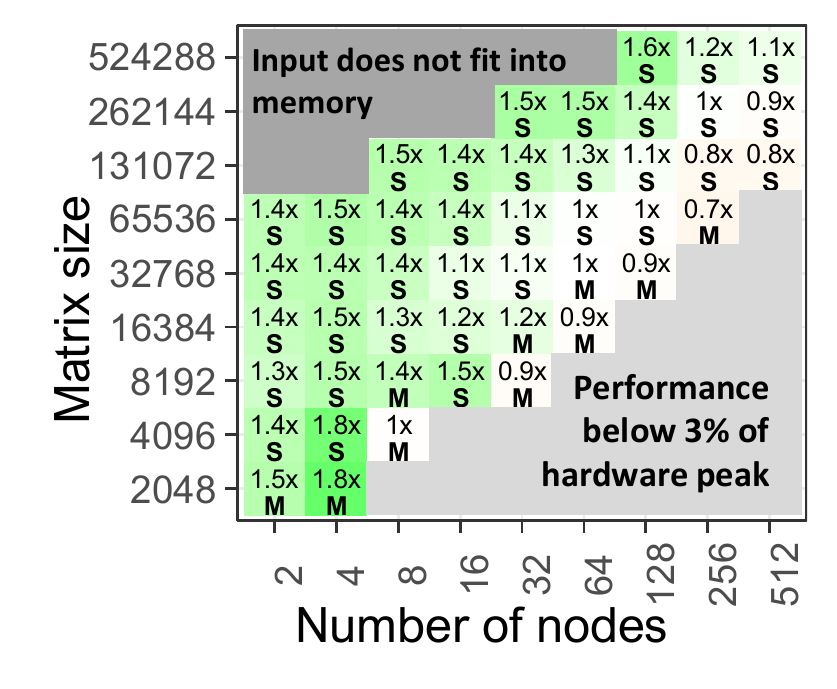}}
	%
	\subfloat
	{		
		\hspace{-1.0em}
		\includegraphics[width=0.234 \textwidth]
		{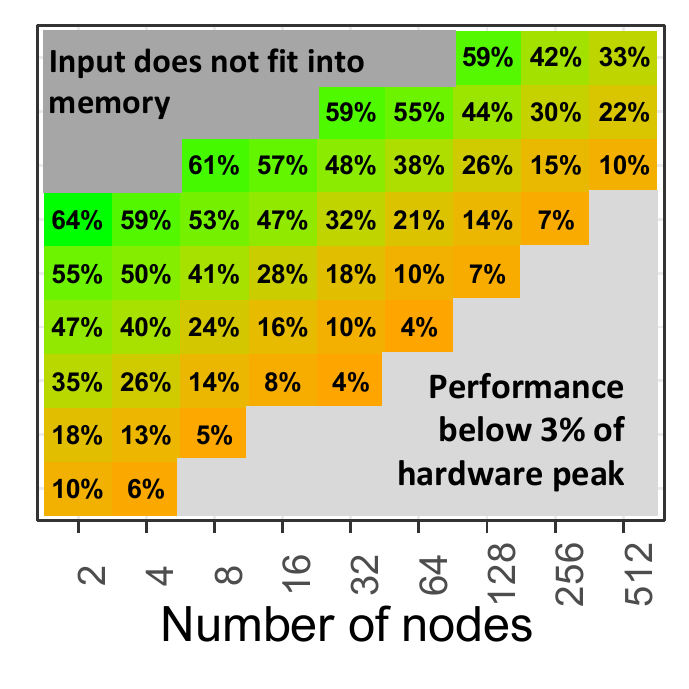}}
	\caption{
		{{\textbf{Left:} measured runtime speedup of \chol vs. 
				fastest state-of-the-art library (S=SLATE~\cite{slate}, 
				C=CAPITAL~\cite{choleskyQRnew}, 
				M=MKL~\cite{mkl}). \textbf{Right:} \chol's achieved \% of 
				machine peak 
				performance. } }
	}
	\label{fig:heatmaps_chol}
\end{figure}

Our experiments confirm advantages of \conflux and \chol in terms of 
both communication volume and time-to-solution over all other implementations 
tested.
A significant communication reduction can be observed (up to 1.42 times for \conflux 
compared with the second-best implementation for $P=$ 1,024). Moreover,
the performance models predict even greater benefits for larger runs
 (expected 
2.1 times communication reduction for a full-machine run on the Summit 
supercomputer -- Figure~\ref{fig:heatmap}). Most importantly, our 
implementations 
consistently outperform existing implementations (up to 
three times
 -- Figures~\ref{fig:heatmaps_lu} and~\ref{fig:performancePlotsLU}).

\noindent\macb{Communication volume.}
Fig.~\ref{fig:square_strong} presents the {measured }communication volume 
per node, as well as our derived cost models (Table{~\ref{tab:comparison}}) 
	presented with solid lines, 
for $N=$ 16,384. Observe that \conflux communicates the least for all values of 
$P$. 
Note that 
since 
both MKL and SLATE use similar 2D decompositions, their communication 
volumes are mostly equal, with a slight advantage for SLATE. 
In 
Fig.~\ref{fig:square_weakp1}, we show the weak scaling characteristics of the 
analyzed implementations. Observe that for a fixed amount of work per node, the 
2D 
algorithms - MKL and SLATE - scale sub-optimally. 
Figure~\ref{fig:heatmap} summarizes the communication volume reduction of 
\conflux compared with the second-best implementation, both for measurements 
and 
theoretical predictions. It can be seen that for all combinations of $P$ and 
$N$, \conflux always communicates the least. For all measured data points, 
the asymptotically optimal CANDMC performed worse than MKL 
or SLATE. 
The figure also presents the predicted communication cost of all considered 
implementations for up to $P=$ 262,144 based on our theoretical models.

\noindent\macb{Performance.} 
Our measurements show that both \conflux and \chol outperform all considered 
state-of-the art libraries in almost all scenarios 
(Figures~\ref{fig:heatmaps_lu} 
and~\ref{fig:heatmaps_chol}). Thanks to the optimized block data decomposition 
and efficient overlap of computation and communication, our implementations 
achieve high performance already on relatively small matrices (approx. 40\% of 
hardware peak for cases where \hltext{\mbox{$N^2/P > 2^{27}$}}). In cases where 
the local domain 
per processor becomes very small (\hltext{\mbox{$N^2 / P < 2^{27}$}}) our block 
decomposition does 
not add that much benefit, since the performance is mostly latency-bound, and 
not bandwidth-bound. \hltext{This is visible not only in strong scaling 
(\mbox{Figures~\ref{fig:performancePlotsLU} 
and~\ref{fig:performancePlotsCholesky}}, \textbf{a)} and \textbf{b)}), but also 
in weak scaling (\textbf{c)}), where the 
input size per processor \mbox{$N^2/P$} is constant. This is again caused by 
latency overheads of scattering data between 1D and 2.5D layouts. }

 However, as the local domains become larger and may be 
more efficiently pipelined and overlapped using asynchronous MPI routines and 
intra-node OpenMP parallelism, the advantage becomes significant 
(Figures~\ref{fig:performancePlotsLU} and~\ref{fig:performancePlotsCholesky}). 
\conflux 
outperforms existing libraries up to three times (for $P=4, N=4096$, 
second-best 
library is SLATE -- Figure~\ref{fig:heatmaps_lu}) and \chol achieves up to 1.8 
times speedup (e.g., $P=4, N= $ 4,096, second-best is again SLATE).

\noindent\macb{Implications for Exascale.} Both the communication models' 
predictions (Figure~\ref{fig:heatmap}) and measured speedups 
(Figures~\ref{fig:heatmaps_lu} and~\ref{fig:heatmaps_chol}) allow us to predict 
that when running our implementations on exascale machines, we can expect to 
see further performance improvements over state-of-the-art libraries. 
Furthermore, throughput-oriented hardware, such as GPUs and 
FPGAs, may benefit even more from the communication reduction of our schedules. 
Thus, \conflux and \chol not only outperform the state-of-the-art 
libraries at relatively small scales --- which are most common use 
cases in practice~\cite{osawa2019large, gb19, joost} --- but also promise 
speedups on full-scale 
performance runs on modern supercomputers.

\section{Related Work}

\begin{table*}[t]
	\vspace{3em}
	\footnotesize
	\begin{tabular}
		{p{1.1cm}p{5cm}p{5.05cm}p{5.15cm}}
		\toprule
		& \textbf{Pebbling}~\cite{sethi1975complete, 
			bruno1976code, 
			redblue, redbluewhite, COSMA} & 
		\textbf{Projection-based}~\cite{general_arrays, 
			demmel2, 
			demmel3, demmel4, 
			ballard2011minimizing, olivry2020automated} 
		& 
		\textbf{Problem specific}~\cite{aggarwal1988input, 
			benabderrahmane2010polyhedral, 
			mehta2014revisiting, darte1999complexity, gb19}
		\\
		\midrule
		\textbf{Scope} & 
		\faThumbsOUp \faThumbsOUp ~General cDAGs & 
		\makecell[tl] {\faThumbsOUp 
			~Programs Geometric
			structure of 
			iteration space}  & 
		\faThumbsDown ~Individually 
		tailored for given problem  \\
		\textbf{Key\newline Features} &
		\makecell[tl] {
			\faThumbsOUp ~General scope \\
			\faThumbsOUp ~Expresses complex data
			dependencies \\
			\faThumbsOUp ~Directly exposes schedules \\
			\faThumbsOUp ~Intuitive \\
			\faThumbsDown ~PSPACE-complete 
			in general case \\
			\faThumbsDown ~No guarantees that a solution  
			exists \\
			\faThumbsDown ~No well-established method
			how to \\
			\hspace{1.4em}automatically translate
			code to cDAGs 
		} 
		& 
		\makecell[tl] {
			\faThumbsOUp ~Well-developed theory and tools  \\
			\faThumbsOUp ~Guaranteed to find solution \\
			\hspace{1.4em}for given class of 
			programs \\
			\faThumbsDown ~Bounds are often not tight \\
			\faThumbsDown ~Fails to capture dependencies\\
			\hspace{1.4em}between statements \\
			\faThumbsDown ~Limited scope
		}
		& 
		\makecell[tl] {
			\faThumbsOUp ~Takes advantage of
			problem-specific  \\
			\hspace{1.4em}features\\
			\faThumbsOUp ~Tends to provide best 
			practical results \\
			\faThumbsDown ~Requires large manual effort \\
			\hspace{1.4em}for each algorithm separately \\
			\faThumbsDown ~Difficult to generalize \\
			\faThumbsDown ~Often based on heuristics\\
			\hspace{1.4em}with no guarantees on optimality 
		}
		\\
		\bottomrule
	\end{tabular}
	\caption{Overview of different approaches to modeling 
		data 
		movement.
	}
	
	\label{tab:stateoftheart}
\end{table*}

Previous 
work on I/O analysis can be categorized into three classes (see 
Table~\ref{tab:stateoftheart}): 
work based on \textbf{direct 
pebbling} or variants of it, such as Vitter's block-based 
model~\cite{vitter1998external}; works using \textbf{geometric arguments 
of projections} based on the Loomis-Whitney 
inequality~\cite{loomisWhitney}; and 
works applying optimizations limited to {specific structural 
properties such as \textbf{affine loops}~\cite{affineloops}, 
and more generally, \textbf{the polyhedral model program 
representation}~\cite{benabderrahmane2010polyhedral, 
	mehta2014revisiting, olivry2020automated}. Although the scopes of those 
	approaches 
significantly overlap --- for example, kernels like matrix multiplication can 
be captured by most of the models --- there are important 
differences both in methodology and the end-results they provide, as summarized 
in 
Table~\ref{tab:stateoftheart}.

Dense linear algebra operators are among the standard core kernels in 
scientific applications.
Ballard et al.{~\cite{ballard2011minimizing}} 
present a comprehensive overview of their asymptotic I/O lower bounds and 
I/O minimizing schedules, both for sparse and dense matrices. Recently, Olivry 
et al. introduced IOLB~\cite{olivry2020automated} --- a framework 
for assessing sequential lower bounds for polyhedral programs. However, their 
computational model disallows recomputation (cf. 
Section~\ref{sec:output_reuse}).

Matrix factorizations are included in most of linear solvers' libraries.
With regard to the parallelization strategy, these libraries may be 
categorized into three groups: \textbf{task-based: } SLATE~\cite{slate} 
(OpenMP tasks), DLAF~\cite{dlaf} (HPX tasks), DPLASMA~\cite{dplasma} (DaGuE 
scheduler), or 
CHAMELEON~\cite{chameleon} (StarPU tasks); \textbf{static 2D parallel: } 
MKL~\cite{mkl}, Elemental~\cite{poulson2013elemental}, or 
Cray 
LibSci~\cite{libsci}; \textbf{communication-minimizing 
2.5D parallel:} CANDMC~\cite{candmc} and CAPITAL~\cite{choleskyQRnew}.
In the last decade, heavy focus was placed on heterogeneous architectures. Most 
GPU vendors offer hardware-customized BLAS 
solvers~\cite{cusolver}. Agullo et. al~\cite{dongarra_gpu_LU} accelerated LU 
factorization using up to 
4 GPUs. Azzam et. al~\cite{dongarra_mixed_precision_LU} utilize NVDIA's GPU 
tensor cores to compute low-precision LU factorization and then iteratively 
refine the linear problem's solution. Moreover, some of the distributed memory 
libraries 
support GPU offloading for local computations~\cite{slate}.

\section{Conclusions}

{In this work, we present a method of analyzing I/O cost of 
DAAP --- a general 
	class of programs that covers many fundamental 
	computational motifs.}
We show, both theoretically and in practice, that our 
pebbling-based approach 
for deriving the I/O lower bounds is \textbf{more general:} 
programs with 
disjoint array accesses cover a wide variety of applications, 
\textbf{more powerful:} it can explicitly capture 
inter-statement dependencies, 
\textbf{more precise:} it 	derives tighter I/O bounds, and 
\textbf{more 
	constructive:} \xpart provides powerful hints for 
	obtaining parallel 
	schedules. 

When applying the approach to LU and Cholesky factorizations, we 
are able to derive 
new 
lower bounds, as well as new, communication-avoiding 
schedules. Not only do they communicate less than state-of-the-art 
2D \textit{and} 3D decompositions --- by a factor of up to 
1.6$\times$ --- but most importantly, they
outperform existing 
commercial libraries in a wide range of problem parameters (up to 3$\times$ for 
LU, up to 1.8$\times$ for Cholesky).
 Finally, our code is openly 
available, offering full ScaLAPACK layout compatibility.

\section{Acknowledgements}
This project received funding from the European Research Council (ERC) under 
the European 
Union’s Horizon
2020 programme (grant agreement DAPP, no. 678880), EPIGRAM-HS project (grant 
agreement no. 801039). Tal Ben-Nun is supported by 
the Swiss 
National Science Foundation (Ambizione Project \#185778). The authors wish to 
 acknowledge the support from the PASC program (Platform for Advanced 
 Scientific Computing), as well as the 
Swiss National Supercomputing Center (CSCS) for providing computing 
infrastructure.

\bibliography{refs}

\clearpage

%
%
%
%
%
%
%
%
%
%
%
%
%
%
%
%
%
%
%
%
%
%
%
%
%

\end{document}